\documentclass[reqno,11pt]{amsart}
\usepackage{fullpage}
\usepackage{amsfonts}
\usepackage{amssymb}
\usepackage{graphicx}
\usepackage{arydshln}
\usepackage{latexsym}
\usepackage[mathscr]{eucal}
\usepackage{cite}
\usepackage{color}
\usepackage{bm}
\usepackage{mathrsfs}
\usepackage{amsmath,amsthm,accents}
\numberwithin{equation}{section}
\DeclareMathOperator{\sech}{sech}

\usepackage{subfigure}
\numberwithin{equation}{section}

\newtheorem{Thm}{Theorem}
\def \wh#1{\widehat{#1}}
\def \wt#1{\widetilde{#1}}
\def \wb#1{\overline{#1}}

\def \dwh{\underaccent{{\cc@style\widehat{\mskip10mu}}}}
\def \dt#1{\underaccent{\tilde}{#1}}

\def \cd#1{\accentset{\circ}{#1}}

\newcommand{\nn}{\nonumber}

\newcommand{\Ga}{\Gamma}

\newcommand{\Ps}{\Psi}

\newcommand{\st}{\hbox{\tiny\it{T}}}

\def \cd#1{\accentset{\circ}{#1}}

\def \d#1{\accentset{\bullet}{#1}}

\begin{document}
	
\title{Discrete nonlinear Schr\"{o}dinger type equations: Solutions and continuum limits}	
\author{Song-lin Zhao$^{*}$,~~Xiao-hui Feng,~~Wei Feng\\
\\\lowercase{\scshape{Department of Applied Mathematics, Zhejiang University of Technology, Hangzhou 310023, P.R. China}}}
\email{*Corresponding Author: songlinzhao@zjut.edu.cn}

\begin{abstract}

As local and nonlocal reductions of a discrete second-order Ablowitz-Kaup-Newell-Segur equation,
two discrete nonlinear Schr\"{o}dinger type equations are considered. Through the bilinearization
reduction method, we construct double Casoratian solutions of the reduced discrete nonlinear
Schr\"{o}dinger type equations, including soliton solutions and Jordan-block solutions.
Dynamics of the obtained one-soliton and two-soliton solutions are analyzed and illustrated. Moreover,
both semi-continuous limit and full continuous limit, are
applied to obtain solutions of the local and nonlocal semi-discrete nonlinear Schr\"{o}dinger type equations, as
well as the local and nonlocal continuous nonlinear Schr\"{o}dinger type equations.

\end{abstract}


\maketitle

\section{Introduction}
\label{sec-1}

The nonlinear Schr\"{o}dinger (NLS) type equations
are classic integrable models emerging from a wide variety of fields and playing a significant
role in many fields of science such as fluids, nonlinear
optics, the theory of deep water waves, plasma physics, nonlinear transmission
networks, nuclear physics, Bose-Einstein condensates, condensed matter physics, and so on (see \cite{Agr,Ben,Has,Mio,Kan,KIV,Fib,Ken}).
The focusing NLS equation, being a standard NLS type equation,
\begin{align}
\label{NLS}
iu_t+u_{xx}+|u|^2u=0
\end{align}
with dependent variable $u:=u(x,t)$ admits carrier waves
and solitons, and whose breathers and other solutions (e.g. rogue waves) are the modulations
of carrier waves, where $i$ is the imaginary unit and $|u|$ stands for the modulus of $u$. In 2013,
Ablowitz and Musslimani proposed a nonlocal $\mathcal {PT}$-symmetric NLS equation \cite{AM-2013-PRL}
\begin{align}
\label{nNLS}
iu_t+u_{xx}+uu^*(-x,-t)u=0,
\end{align}
bringing a fresh look in the research of integrable systems. Analogous to \eqref{NLS},
the nonlocal NLS equation \eqref{nNLS} also exhibits lots of interesting properties. For instance,
this equation can be also solved by the inverse scattering transform, have multi-soliton solutions,
admit linear Lax pair formulation and possess an infinite number of conservation laws \cite{Ab-nNLS-IST}.
In recent years, many approaches, including
inverse scattering transformation, Darboux transformation, Hirota's bilinear method, the Riemann-Hilbert approach,
Cauchy matrix approach and the assumption of stationary solution, were applied to construct various solutions to the equation
\eqref{nNLS}, such as solitons, rational solutions, Jacobi elliptic-function and hyperbolic-function solutions, and so on
\cite{Val,Sin,Wen,HL-EPJP,CZ-AML,FZ-ROMP,XCLM,AbLM}. Not only the exact solutions,
great other progress on the equation \eqref{nNLS} has also been achieved. In \cite{Ger} and \cite{Ryb-LTA}, complete integrability
and long-time asymptotics behavior were investigated, respectively. Subsequently, the application of the nonlocal
equation \eqref{nNLS} on multi-place physics was discussed \cite{Lou-CTP}. In a recent paper \cite{AbM-shifted}, Ablowitz and
Musslimani proposed several novel integrable nonlocal reductions for
the AKNS scattering problem and generalized the equation \eqref{nNLS} into a shifted form.

The discrete integrable systems always serve as master models in the theory of
integrable systems, due to their rich algebraic structure and connection with modern mathematics and physics, such as
difference calculus, elliptic functions, numerical calculation, statistical physics, subatomic physics, and so on,
revealing on frequent occasions unexpected links between them (see \cite{HJN-book-2016} and the references therein).
The discrete integrable systems pioneered at the end of the 1970's by Ablowitz and Ladik \cite{AL,AL-1}, and Hirota \cite{Hir},
are consolidated throughout the 1980's by the work of many others \cite{DJM,NQC}, have got great progress
\cite{NW,ABS,Nijhoff-CM-2009,HZ-JPA,IST,KP,BSQ,ZZN-BSQ,ZZ-SAM-2013}.
Recently, a great deal of interest has been directed at the creation and study of discrete models of the NLS type equations.
As the associated nonlinear superposition principle of auto-B\"{a}cklund transformation for NLS hierarchy,
the first system of discretization of NLS equation appeared in the paper \cite{Kono} by Konopelchenko
put forward a discretization theory of NLS type equations. This discrete model
was rediscovered by Hattori and Willox through symmetry-constraint method from Hirota-Miwa equation \cite{WH}, which
possesses an infinite set of polynomial higher symmetries and rational higher symmetries
in each lattice direction \cite{Tsuch}. Using the discretized bilinear identity,
Date, Jimbo and Miwa \cite{DJM-JPSJ} proposed an integrable ``fake''
discrete NLS equation since the resulting equation does not admit the complex
conjugation reduction between the dependent functions.
In \cite{QNCL}, a discrete NLS type Nijhoff-Quispel-Caple equation as well as its semi-discrete
model was established based on the direct linearization framework \cite{FA-DLM}.
Recently, a discrete NLS equation and a modified discrete NLS equation \cite{ZFJ-CNSNS}
were constructed with the help of the Cauchy matrix approach \cite{Nijhoff-CM-2009,ZZ-SAM-2013}.
where some exact solutions including soliton solutions, Jordan-block solutions and mixed solutions for the resulting discrete equations were generated
by solving the determining equation set.

Apart from the above-mentioned discrete NLS type models, there is also an interesting discrete NLS equation
\begin{subequations}
\label{dNLS}
\begin{align}
& i(\wh{u}-u)-\delta(\wh{u}+u)+\delta(1+|u|^2)(\wt{u}+\wh{\dt{u}})w=0,\\
& (1+|\wh{u}|^2)\wt{w}=(1+|u|^2)w,
\end{align}
\end{subequations}
introduced by Horiuchi {\it et al.} \cite{HOHT} in terms of the framework of
the integrable discretization based on bilinear forms, which was reduced to
the semi-discrete NLS equation \cite{AL-1975} in the continuum limit of $\delta$.
In \eqref{dNLS}, the dependent variable $u:=u_{n,m}$ is a function of discrete
variables $\{n, m\}$ associated and $\delta$ is a continuous lattice parameter. The
notations $\wt{\phantom{a}}$ and $\wh{\phantom{a}}$ represent the discrete forward shift operators
and notation $\dt{\phantom{a}}$ means the backward shift operator, e.g.
$\wt{u}=u_{n+1,m}$, $\wh{u}=u_{n,m+1}$ and $\dt{\wh{u}}=u_{n-1,m+1}$.

Although bright multi-soliton solutions in Casoratian form for the discrete NLS equation \eqref{dNLS} have been obtained by
Tsujimoto in his Ph.D. thesis \cite{Tsujimoto}, we will pay attention to the explicit one-, two-soliton solutions, as well as their
dynamical behaviours. In addition, we are also interested in the Jordan-block solutions for this discrete equation.
In order to investigate double Casoratian solutions of the discrete NLS equation \eqref{dNLS}, we introduce a coupled discrete system
\begin{subequations}
\label{dAKNS2}
\begin{align}
& i(\wh{u}-u)-\delta(\wh{u}+u)+\delta(1+uv)(\wt{u}+\wh{\dt{u}})w=0, \\
& i(\wh{v}-v)+\delta(\wh{v}+v)-\delta(1+uv)(\wt{v}+\wh{\dt{v}})w=0, \\
& (1+\wh{u}\wh{v})\wt{w}=(1+uv)w,
\end{align}
\end{subequations}
where $w$ is an auxiliary variable expressed as
\begin{align}
\label{dAKNS-c-rw}
w_{n,m}=\prod_{j=n_0}^{n-1}\frac{1+u_{j,m}v_{j,m}}{1+\wh u_{j,m}\wh v_{j,m}}, \quad n_0\in\mathbb{Z}.
\end{align}
For the sake of simplicity, we name \eqref{dAKNS2} by dAKNS(2) in the rest part of the present paper.
The discrete NLS equation \eqref{dNLS} can be deduced from system \eqref{dAKNS2} by imposing
local constraint $(v=u^*, w=w^*)$, where and whereafter the asterisk denotes complex conjugation.
The discrete system \eqref{dAKNS2} admit a nonlocal reduction. In fact, through the
constraint $(v=u_{-n,-m},~w=\wh{w}_{-n,-m})$, we arrive at a discrete equation
\begin{subequations}
\label{ndNLS}
\begin{align}
& i(\wh{u}-u)-\delta(\wh{u}+u)+\delta(1+u u_{-n,-m})(\wt{u}+\wh{\dt{u}})w=0,\\
& (1+\wh{u}\wh{u}_{-n,-m})\wt{w}=(1+u u_{-n,-m})w,
\end{align}
\end{subequations}
which is referred to as a nonlocal discrete NLS equation
because of the simultaneous involvement of $u$ and $u_{-n,-m}$. In what follows,
we call \eqref{dNLS} and \eqref{ndNLS}, dNLS equation and ndNLS equation for short, respectively.

Because, as stressed above, the discrete integrable systems have played an increasingly
prominent role in both mathematics and physics during past decades, the nonlocal
discrete integrable systems offer great potential for building media of various areas of the integrable systems.
The first remarkable work is the proposal of nonlocal Adler-Bobenko-Suris
lattice equations \cite{ZKZ-CAC}. With regards to the exact solutions to nonlocal discrete integrable systems,
one can refer to references \cite{XFZ-TMP,XZS-SAPM} and a recent paper \cite{ZXS-MMAS}
for Cauchy matrix and bilinearization reduction scheme
to the nonlocal discrete sine-Gordon equation and nonlocal discrete modified Korteweg-de Vries equation.
The bilinearization reduction scheme, which ultimately aim to construct exact solutions in
double Wronskian/Casoratian form for the reduced equations, is naturally related to the solutions
for the double Wronskian/Casoratian solutions of the before-reduction system.
This method involves, first taking appropriate reductions to get the local/nonlocal integrable systems,
and second solving the matrix equation algebraically to derive the exact solutions.
Compared with the Cauchy matrix reduction approach \cite{XFZ-TMP},
there is a great advantage of the bilinearization reduction method. The latter method allows one
to construct exact solutions of the real reduced equations. While in the Cauchy matrix reduction
scheme, one cannot achieve that.

In this paper, we focus on the double Casoratian solutions of the dNLS equation \eqref{dNLS} and ndNLS equation
\eqref{ndNLS}, as well as the continuum limits.
The paper is arranged as follows: In Section 2, double Casoratian solutions for system \eqref{dAKNS2} are discussed.
In Sections 3 and 4, we would like to construct various of exact solutions for the dNLS equation \eqref{dNLS} and
ndNLS equation \eqref{ndNLS} by using the bilinearization reduction scheme, respectively.
Some exact solutions, including multi-solitons and Jordan-block solutions, as well as
the dynamics are investigated. Moreover, continuum limits, including semi-continuous limit and continuous limit,
are introduced to discuss the semi-discrete NLS-type equations, as well as the corresponding exact solutions.
Section 5 is for conclusions and some remarks.

\section{Double Casoratian solutions to the dAKNS(2) equation \eqref{dAKNS2}}

Our aim in this section is to construct double Casoratian solutions to the dAKNS(2) equation \eqref{dAKNS2}.
We first introduce some notations on the double Casoratian. Let $\Phi$ and $\Psi$ be $(N+M+2)$-th order column vectors
\[\Phi=(\Phi_1, \Phi_2,\cdots,\Phi_{N+M+2})^{\st}, \quad \Psi=(\Psi_1, \Psi_2,\cdots,\Psi_{N+M+2})^{\st},\]
where entries $\Phi_j$ and $\Psi_j$ are functions of $(n,m)$. A standard double Casoratian is a determinant
of the following form:
\begin{align}
\mbox{C}^{(N,M)}(\Phi, \Psi)=|\Phi, E^2\Phi, \ldots, E^{2N}\Phi; \Psi, E^2\Psi, \ldots, E^{2M}\Psi|,
\end{align}
where $E$ is the shift operator defined as $E^jf_{n,m}=f_{n+j,m}$. We introduce short-hand \cite{FN-1983}
\begin{align}
\mbox{C}^{(N,M)}(\Phi, \Psi)=|\wh{\Phi^{(N)}};\wh{\Psi^{(M)}}|,
\end{align}
where by $\wh{\Phi^{(N)}}$ we mean consecutive columns $(\Phi, E^2\Phi, \ldots, E^{2N}\Phi)$.

Inserting transformation of dependent variables
\begin{align}
\label{dAKNS2-uvw-tran}
u=g/f,\quad v=h/f,\quad w=\wh{\dt{f}}f/(\wh{f}\dt{f}),
\end{align}
in system \eqref{dAKNS2}, after straightforward manipulations one arrives at the following bilinear form of $f$, $g$ and $h$,
\begin{subequations}
\label{dAKNS2-bili}
\begin{align}
& i(\wh{g}f-g\wh{f})-\delta(\wh{g}f+g\wh{f})+\delta(\wt{g}\wh{\dt{f}}+\wh{\dt{g}}\wt{f})=0,\\
& i(\wh{h}f-h\wh{f})+\delta(\wh{h}f+h\wh{f})-\delta(\wt{h}\wh{\dt{f}}+\wh{\dt{h}}\wt{f})=0,\\
& \wt{f}\dt{f}-f^2=gh.
\end{align}
\end{subequations}
System \eqref{dAKNS2-bili} is in Hirota bilinear form since it is gauge invariant under transformation
\begin{align}
f\rightarrow f~\mbox{exp}(\alpha_0 n+\beta_0 m), \quad
g\rightarrow g~\mbox{exp}(\alpha_0 n+\beta_0 m), \quad
h\rightarrow h~\mbox{exp}(\alpha_0 n+\beta_0 m),
\end{align}
where $\alpha_0$ and $\beta_0$ are two constants \cite{HZ-JPA}.

We restrict our attention to the following condition equations set (CES)
\begin{subequations}
\label{dAKNS2-CES}
\begin{align}
& \wt{\Phi}=e^K\Phi, \quad \wh{\Phi}=[(\delta(E^{2}-1)-i)(\delta(1-E^{-2})-i)^{-1}]^{\frac{1}{2}}\Phi, \\
& \wt{\Psi}=e^{-K}\Psi, \quad \wh{\Psi}=[(\delta(E^{2}-1)-i)(\delta(1-E^{-2})-i)^{-1}]^{\frac{1}{2}}\Psi,
\end{align}
\end{subequations}
where $K$ is an invertible constant matrix of order $N+M+2$.
For the double Casoratian solutions of bilinear system \eqref{dAKNS2-bili} we have the following result.
\begin{Thm}
The double Casorati determinants
\begin{align}
\label{dAKNS2-dCs}
f=|\wh{\Phi^{(N)}};\wh{\Psi^{(M)}}|,\quad g=|\wh{\Phi^{(N+1)}};\wh{\Psi^{(M-1)}}|, \quad h=|\wh{\Phi^{(N-1)}};\wh{\Psi^{(M+1)}}|
\end{align}	
solve the bilinear system \eqref{dAKNS2-bili}, provided that the basic column vectors $\Phi$ and $\Psi$ are given by the CES \eqref{dAKNS2-CES}.
\end{Thm}

The CES \eqref{dAKNS2-CES} implies
\begin{subequations}
\label{dAKNS2-Phsi-K}
\begin{align}
& \Phi=e^{Kn}[(2\delta e^{K}\sinh K-iI)(2\delta e^{-K}\sinh K-iI)^{-1}]^{\frac{m}{2}}C^{+}, \\
& \Psi=e^{-Kn}[(2\delta e^{K}\sinh K-iI)(2\delta e^{-K}\sinh K-iI)^{-1}]^{-\frac{m}{2}}C^{-},
\end{align}
\end{subequations}
where $C^{\pm}=(c_{1}^{\pm},c_{2}^{\pm},\dots c_{N+1}^{\pm};d_{1}^{\pm},d_{2}^{\pm},\dots d_{M+1}^{\pm})^{\st}$ are constant column vectors
and $I$ is the unit matrix whose index indicating the size is omitted.

Analogous to the similar analysis in \cite{XZS-SAPM},  we replace $K$ by its similar matrix, giving the following result.
\begin{Thm}
\label{Solu-K-wbK}
$K$ and its any similar form lead to the same $u,~v$ and $w$ through \eqref{dAKNS2-uvw-tran} and \eqref{dAKNS2-dCs}.
\end{Thm}
\begin{proof}
Let $\wb{K}$ be a matrix that is similar to $K$, i.e.,
\begin{align}
\label{Ga-K}
\wb{K}=TKT^{-1},
\end{align}
where $T$ is the transform matrix. Substituting \eqref{Ga-K} into \eqref{dAKNS2-Phsi-K} yields two new basic column vectors
\begin{subequations}
\label{dKANS2-Phsi-wbK}
\begin{align}
& \wb{\Phi}=T\Phi=e^{\wb{K} n}[(2\delta e^{\wb{K}}\sinh \wb{K}-iI)(2\delta e^{-\wb{K}}\sinh \wb{K}-iI)^{-1}]^{\frac{m}{2}}\wb{C}^{+}, \\
& \wb{\Psi}=T\Ps=e^{-\wb{K} n}[(2\delta e^{\wb{K}}\sinh \wb{K}-iI)(2\delta e^{-\wb{K}}\sinh \wb{K}-iI)^{-1}]^{-\frac{m}{2}}\wb{C}^{-}
\end{align}
\end{subequations}
with $\wb{C}^{\pm}=TC^{\pm}$. Denoting $f(\Phi,\Psi)$, $g(\Phi,\Psi)$ and
$h(\Phi,\Psi)$ are the double Casorati determinants composed by $\Phi$
and $\Psi$ defined as \eqref{dAKNS2-Phsi-K}, then we can easily find
$f(\wb{\Phi},\wb{\Psi})=|T|f(\Phi,\Psi)$, $g(\wb{\Phi},\wb{\Psi})=|T|g(\Phi,\Psi)$
and $h(\wb{\Phi},\wb{\Psi})=|T|h(\Phi,\Psi)$, which indicates that
$u,~v$ and $w$ are similarity invariant for $K$ and its any similar matrix.
\end{proof}

By the Theorem \ref{Solu-K-wbK}, in what follows we just need to consider the Jordan canonical form of $\Phi$ and $\Psi$:
\begin{subequations}
\label{dAKNS2-Phsi-Ga}
\begin{align}
\label{dAKNS2-Phi-Ga}
& \Phi=e^{\Ga n}[(2\delta e^{\Ga}\sinh \Ga-iI)(2\delta e^{-\Ga}\sinh \Ga-iI)^{-1}]^{\frac{m}{2}}C^{+}, \\
& \Psi=e^{-\Ga n}[(2\delta e^{\Ga}\sinh \Ga-iI)(2\delta e^{-\Ga}\sinh \Ga-iI)^{-1}]^{-\frac{m}{2}}C^{-},
\end{align}
\end{subequations}
where $\Ga$ is the Jordan canonical matrix.

\section{dNLS equation \eqref{dNLS}: solutions and continuum limits}

In this section, we perform the framework of bilinearization reduction to
study exact solutions of the equation \eqref{dNLS} from those of the dAKNS(2) system \eqref{dAKNS2}.
Moreover, semi-continuous limit and full-continuous limit are introduced
to produce the semi-discrete NLS equation, as well as the continuous NLS equation, respectively.
Exact solutions of the resulting semi-discrete NLS equation are also discussed.

\subsection{Solutions to equation \eqref{dNLS}}

The dNLS equation \eqref{dNLS} is reduced from the system \eqref{dAKNS2} by imposing constraint $(v=u^*,w=w^*)$.
After taking $M=N$ and replacing $C^{\pm}$ by $e^{\mp N\Gamma}C^{\pm}$, equation \eqref{dAKNS2-dCs} gives the
solutions of \eqref{dNLS}, which are listed in the next theorem.
\begin{Thm}
The functions
\begin{align}
\label{dNLS-uw-fg}
u=g/f,\quad w=\wh{\dt{f}}f/(\wh{f}\dt{f})
\end{align}
with
\begin{align}
\label{dNLS-solu}
f=|e^{-N\Ga}\wh{\Phi^{(N)}};
e^{N\Ga}\wh{\Psi^{(N)}}|,\quad
g=|e^{-N\Ga}\wh{\Phi^{(N+1)}};
e^{N\Ga}\wh{\Psi^{(N-1)}}|
\end{align}
solve the equation \eqref{dNLS}, if the $(2N+2)$-th order column vectors $\Phi$ and $\Psi$ are given
by \eqref{dAKNS2-Phsi-Ga} satisfy the relation
\begin{align}
\label{dNLS-Phsi-T}
\Psi=T\Phi^*,
\end{align}
where $T\in\mathbb{C}^{(2N+2)\times(2N+2)}$ is a constant matrix satisfying matrix equations
\begin{align}
\label{dNLS-AT}
\Ga T+ T\Ga^*=\bm 0,\quad TT^*=-I,
\end{align}
and we require $C^-=TC^{+^*}$.
\end{Thm}

\noindent {\bf Remark 1.} \textit{The substitution of \eqref{dNLS-Phsi-T} into \eqref{dNLS-solu} yields that
functions \eqref{dNLS-uw-fg} together with
\begin{align}
\label{dNLS-solu-1}
f=|e^{-N\Ga}\wh{\Phi^{(N)}};
e^{N\Ga}T\wh{\Phi^{*(-N)}}|,\quad
g=|e^{-N\Ga}\wh{\Phi^{(N+1)}};
e^{N\Ga}T\wh{\Phi^{*(-N+1)}}|,
\end{align}
solve the dNLS equation \eqref{dNLS}.}

Next we give some explicit solutions for the dNLS equation \eqref{dNLS}. As a preliminary step, let matrices $\Ga$ and $T$ be of form
\begin{align}
\label{dNLS-LT-form}
\Ga=\left(
\begin{array}{cc}
L & \bm 0  \\
\bm 0 & -L^*
\end{array}\right),
\quad
T=\left(
\begin{array}{cc}
\bm 0 & I  \\
-I & \bm0
\end{array}\right),
\end{align}	
where $L$ is a $(N+1)$-th Jordan canonical matrix. Soliton solutions
and Jordan-block solutions can be derived in view of different kinds of eigenvalues of $\Ga$.

\subsubsection{Soliton solutions}
If $L$ is a diagonal matrix
\begin{align}
\label{L-diag}
L=\text{diag}(k_1,k_2,\ldots, k_{N+1}),
\end{align}
where complex constants $\{k_j|k_i\neq k_j,~i\neq j\}$ distinguish the discrete spectral parameters.
Then we immediately get the entries in the Casoratian \eqref{dNLS-solu-1}, given by
\begin{align}
\label{dNLS-Phij-ex}
\Phi_{j}=\left\{
\begin{aligned}
& c_je^{\xi_j}, \quad j=1, 2, \ldots, N+1, \\
& d_je^{-\xi_s^*}, \quad j=N+1+s,\quad s=1,2,\dots,N+1,
\end{aligned}	
\right.
\end{align}
where $c_j=c_j^+$ and $d_j=d_j^+$ are complex constants and
\begin{align}
\label{xi}
\xi_j=k_jn+\tau_j m, \quad e^{\tau_j}=\bigg(
\dfrac{i-2\delta e^{k_j}\sinh k_j}{i-2\delta e^{-k_j}\sinh k_j}\bigg)^{\frac{1}{2}}, \quad j=1,2,\ldots,N+1.
\end{align}	

In the case of $N=0$, we arrive at the one-soliton solution
\begin{subequations}
\label{dNLS-1ss}
\begin{align}
\label{dNLS-1ss-u}
& u=\frac{2c_1d_1e^{2i\text{Im}(k_1+\xi_1)}\sinh(2\text{Re}k_1)}
{|c_1|^2e^{2\text{Re}\xi_1}+|d_1|^2e^{-2\text{Re}\xi_1}}, \\
& w=\wh{\dt{f}}f/(\wh{f}\dt{f}), \quad f=-(|c_1|^2e^{2\text{Re}\xi_1}+|d_1|^2e^{-2\text{Re}\xi_1}),
\end{align}
\end{subequations}
where $\text{Re}(\kappa)$ and $\text{Im}(\kappa)$ represent real and
imaginary parts of $\kappa$, respectively. To consider the dynamics of this solution,
we denote $k_1=(a_1+ib_1)/2$ and get the envelop
\begin{align}
\label{dNLS-1ss-u-mo}
|u|^2=\sinh^2a_1\sech^2\big(a_1n+m\ln\rho_1/2+\ln|c_1/d_1|\big),
\end{align}
where
\begin{align}
\label{rho}
\rho_1=\frac{\delta^2{(e^{a_1}\cos{b_1}-1)}^2+{(\delta e^{a_1}\sin b_1-1)}^2}{\delta^2{(e^{-a_1}\cos{b_1}-1)}^2+{(\delta e^{-a_1}\sin{b_1}-1)}^2}.
\end{align}
It is easily seen that solution \eqref{dNLS-1ss-u-mo} is a stationary wave
when $\rho_1=1$. While when $\rho_1\neq 1$, this
solution describes a unidirectional soliton
traveling with amplitude $\sinh^2 a_1$, initial phase $\ln|c_1/d_1|$,
velocity $-\ln\rho_1/(2a_1)$, and trajectory (top trace)
\[n(m)=-(m\ln\rho_1/2+\ln|c_1/d_1|)/a_1.\]
We depict such a wave in Figure 1.
\vskip20pt
\begin{center}
\begin{picture}(120,80)
\put(-160,-23){\resizebox{!}{3.5cm}{\includegraphics{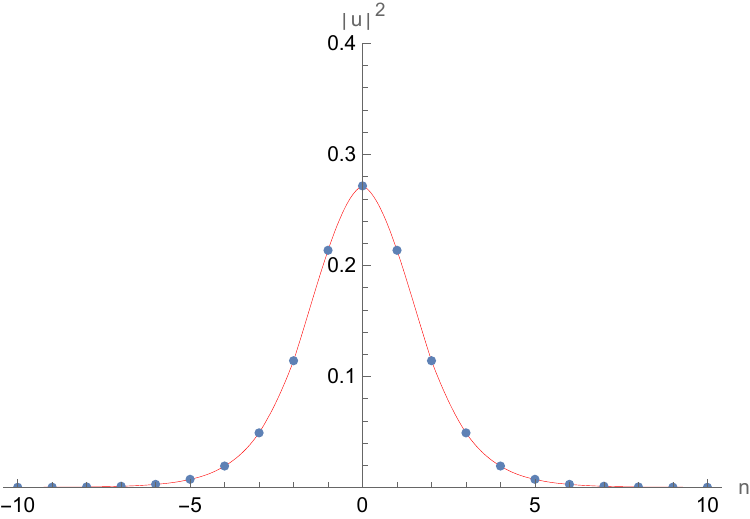}}}
\put(-20,-23){\resizebox{!}{4.0cm}{\includegraphics{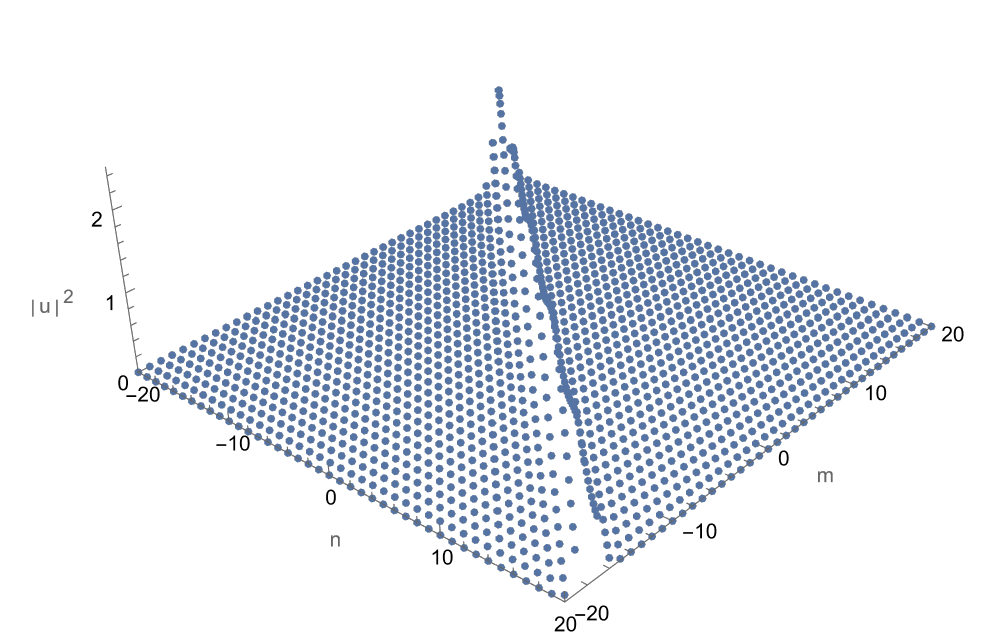}}}
\put(160,-23){\resizebox{!}{3.5cm}{\includegraphics{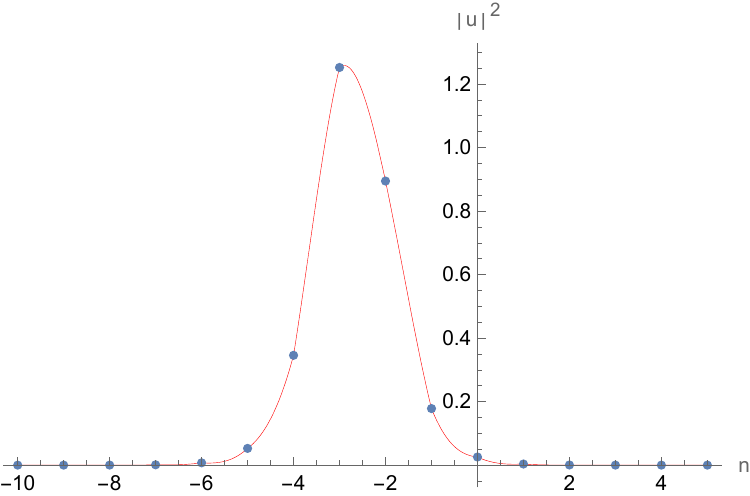}}}
\end{picture}
\end{center}
\vskip 20pt
\begin{center}
\begin{minipage}{15cm}{\footnotesize
\qquad\qquad\qquad(a)\qquad\qquad\qquad\qquad\qquad\qquad\qquad\qquad\quad(b)\qquad\qquad\qquad\qquad\qquad\qquad\qquad\quad (c) \\
{\bf Figure 1}. One-soliton solution $|u|^2$ given by \eqref{dNLS-1ss-u-mo} with $c_1=d_1=1$:
(a) a stationary wave with $a_1=0.5$, $\delta=1$ and corresponding $b_1$;
(b) shape and movement with $k_1=0.5+0.4i$ and $\delta=-2$;
(c) 2D-plot of (b) at $m=3$.}
\end{minipage}
\end{center}

In the case of $N=1$, we get two-soliton solution \eqref{dNLS-uw-fg} with
\begin{subequations}
\label{dNLS-2ss-fg-xi}
\begin{align}
f=\begin{vmatrix}
   c_1e^{-k_1+\xi_1} & c_1e^{k_1+\xi_1} &  d_1^*e^{k_1-\xi_1} & d_1^*e^{-k_1-\xi_1}  \\
   c_2e^{-k_2+\xi_2} & c_2e^{k_2+\xi_2} & d_2^*e^{k_2-\xi_2} & d_2^*e^{-k_2-\xi_2}  \\
   d_1e^{k_1^*-\xi_1^*} & d_1e^{-k_1^*-\xi_1^*} & -c_1^*e^{-k_1^*+\xi_1^*} & -c_1^*e^{k_1^*+\xi_1^*}  \\
   d_2e^{k_2^*-\xi_2^*} & d_2e^{-k_2^*-\xi_2^*} & -c_2^*e^{-k_2^*+\xi_2^*} & -c_2^*e^{k_2^*+\xi_2^*} \\
\end{vmatrix},
\end{align}
\begin{align}
g=\begin{vmatrix}
   c_1e^{-k_1+\xi_1} & c_1e^{k_1+\xi_1} & c_1e^{3k_1+\xi_1} & d_1^*e^{k_1-\xi_1}  \\
   c_2e^{-k_2+\xi_2} & c_2e^{k_2+\xi_2} & c_2e^{3k_2+\xi_2} & d_2^*e^{k_2-\xi_2}  \\
   d_1e^{k_1^*-\xi_1^*} & d_1e^{-k_1^*-\xi_1^*} & d_1e^{-3k_1^*-\xi_1^*} & -c_1^*e^{-k_1^*+\xi_1^*}  \\
   d_2e^{k_2^*-\xi_2^*} & d_2e^{-k_2^*-\xi_2^*} & d_2e^{-3k_2^*-\xi_2^*} & -c_2^*e^{-k_2^*+\xi_2^*}  \\
\end{vmatrix}.
\end{align}
\end{subequations}

Now, let us analyze the interaction of two-soliton. For convenience, the
asymptotic solitons are called as $k_1$-soliton and $k_2$-soliton, respectively. We rewrite $\xi_j,~(j=1,2)$ as
\begin{align}
\label{dNLS-xij}
\xi_j=\mathcal{A}_j+i\mathcal{B}_j, \quad \text{with} \quad
\mathcal{A}_j=(2a_jn+m\ln\rho_j)/4 \quad \text{and} \quad \mathcal{B}_j=(b_jn+m\arg e^{2\tau_j})/2,
\end{align}
in which $\rho_2=\rho_1|_{a_1\rightarrow a_2,b_1\rightarrow b_2}$ and
\begin{align}
\arg e^{2\tau_j}=\arctan\frac{\bar{X}_{a_j}\cd{X}_{-a_j}+\bar{X}_{-a_j}\cd{X}_{a_j}}
{\bar{X}_{a_j}\bar{X}_{-a_j}-\cd{X}_{a_j} \cd{X}_{-a_j}},
\end{align}
where
\begin{align}
\label{dNLS-Xa12}
\bar{X}_{a_j}:=\bar{X}^{b_j}_{a_j}=\delta(1-e^{a_j}\cos b_j), \quad \cd{X}_{a_j}:=\cd{X}^{b_j}_{a_j}=1-\delta e^{a_j}\sin b_j,
\end{align}
and $\bar{X}_{-a_j}$ and $\cd{X}_{-a_j}$ are the cases of replacing $a_j$ with $-a_j$ inside $\bar{X}_{a_j}$ and $\cd{X}_{a_j}$, respectively.
It is easy to find that $e^{\xi_j}\rightarrow \infty$ as
$\mathcal{A}_j\rightarrow +\infty$ and $e^{\xi_j}\rightarrow 0$ as
$\mathcal{A}_j\rightarrow -\infty$. Without loss of generality, we assume $c_j=d_j=1$, $0<a_1<a_2$ and $0<b_1<b_2<\pi/2$ to
guarantee $\mathcal{X}=a_1\ln\rho_2-a_2\ln\rho_1>0$.

For fixed $\xi_1$, it follows from \eqref{dNLS-xij} that
\begin{align}
\label{dNLS-xi2-1}
\xi_2=a_2\xi_1/a_1+m\mathcal{X}/(4a_1)+i(a_1\mathcal{B}_2-a_2\mathcal{B}_1)/a_1,
\end{align}
satisfies $e^{\xi_2}\rightarrow \infty$ as
$m\rightarrow +\infty$ and $e^{\xi_2}\rightarrow 0$ as $m\rightarrow -\infty$.
When $m\rightarrow +\infty$, neglecting subdominant exponential terms, we have
\begin{subequations}
\label{dNLS-fg-xi1+}
\begin{align}
&f\simeq4e^{-2\text{Re}(\xi_1-\xi_2)}(|\sinh(k_1+k^*_2)|^2+e^{4\text{Re}\xi_1}|\sinh(k_1-k_2)|^2),\\
&g\simeq8e^{2\text{Re}\xi_2+2i\text{Im}(k_1+k_2+\xi_1)}\sinh a_1\sinh(k_1-k_2)\sinh(k^*_1+k_2),
\end{align}	
\end{subequations}
and thus the $k_1$-soliton appears
\begin{align}
\label{dNLS-u-xi1+}
u\simeq\dfrac{2e^{2\xi_1+2i\text{Im}(k_1+k_2)}\sinh a_1\sinh(k_1-k_2)
\sinh(k^*_1+k_2)}{|\sinh(k_1+k^*_2)|^2+e^{4\text{Re}\xi_1}|\sinh(k_1-k_2)|^2}.
\end{align}
When $m\rightarrow -\infty$, we find
\begin{subequations}
\label{dNLS-fg-xi1-}
\begin{align}
&f\simeq4e^{-2\text{Re}(\xi_1+\xi_2)}(|\sinh(k_1-k_2)|^2+e^{4\text{Re}\xi_1}|\sinh(k_1+k^*_2)|^2),\\
&g\simeq8e^{-2\text{Re}\xi_2+2i\text{Im}(k_1+k_2+\xi_1)}\sinh a_1\sinh(k^*_1-k^*_2)\sinh(k_1+k^*_2),
\end{align}	
\end{subequations}
and thus the $k_1$-soliton exhibits
\begin{align}
\label{dNLS-u-xi1-}
u\simeq\dfrac{2e^{2\xi_1+2i\text{Im}(k_1+k_2)}\sinh a_1\sinh(k^*_1-k^*_2)
\sinh(k_1+k^*_2)}{|\sinh(k_1-k_2)|^2+e^{4\text{Re}\xi_1}|\sinh(k_1+k^*_2)|^2}.
\end{align}
As a results, the envelops of \eqref{dNLS-u-xi1+} and \eqref{dNLS-u-xi1-} can be expressed as
\begin{align}
\label{dNLS-u-xi1-mo}
|u|^2\simeq\dfrac{4\sinh^2 a_1(\cosh a^{-}_{12}-\cos b^{-}_{12})(\cosh a^{+}_{12}-\cos b^{-}_{12})
e^{2a_1n}\rho_1^m}
{(\cosh a^{\pm}_{12}-\cos b^{-}_{12}+(\cosh a^{\mp}_{12}-\cos b^{-}_{12})e^{2a_1n}\rho_1^{m})^2}, \quad
m\rightarrow \pm\infty,
\end{align}
where $\cdot_{ij}^{\pm}=\cdot_{i}\pm \cdot_j$, which travel with amplitude $\sinh^2a_1$,
velocity $-\ln\rho_1/(2a_1)$, and trajectory
\begin{align}
n(m)=-\frac{1}{2a_1}\left(m\ln\rho_1\pm
\ln\left(\frac{\cosh a^{-}_{12}-\cos b^{-}_{12}}{\cosh a^{+}_{12}-\cos b^{-}_{12}}\right)\right), \quad m\rightarrow \pm \infty.
\end{align}
The phase shift is consequently expressed as
\begin{align}
n(m)=-\frac{1}{a_1}
\ln\left(\frac{\cosh a^{-}_{12}-\cos b^{-}_{12}}{\cosh a^{+}_{12}-\cos b^{-}_{12}}\right).
\end{align}

For fixed $\xi_2$, it follows from \eqref{dNLS-xij} that
\begin{align}
\label{dNLS-xi1-2}
\xi_1=a_1\xi_2/a_2-m\mathcal{X}/(4a_2)+i(a_2\mathcal{B}_1-a_1\mathcal{B}_2)/a_2,
\end{align}
satisfies $e^{\xi_1}\rightarrow 0$ as
$m\rightarrow +\infty$ and $e^{\xi_1}\rightarrow \infty$ as $m\rightarrow -\infty$.
When $m\rightarrow +\infty$, neglecting subdominant exponential terms we derive
\begin{subequations}
\label{dNLS-fg-xi2+}
\begin{align}
&f\simeq4e^{-2\text{Re}(\xi_1+\xi_2)}(|\sinh(k_1-k_2)|^2+e^{4\text{Re}\xi_2}|\sinh(k_1+k^*_2)|^2),\\
&g\simeq-8e^{-2\text{Re}\xi_1+2i\text{Im}(k_1+k_2+\xi_2)}\sinh a_2\sinh(k^*_1-k^*_2)\sinh(k^*_1+k_2),
\end{align}	
\end{subequations}
the $k_2$-soliton thereby reads as
\begin{align}
\label{dNLS-u-xi2+}
u\simeq-\dfrac{2e^{2\xi_2+2i\text{Im}(k_1+k_2)}\sinh a_2\sinh(k^*_1-k^*_2)
\sinh(k^*_1+k_2)}{|\sinh(k_1-k_2)|^2+e^{4\text{Re}\xi_2}|\sinh(k_1+k^*_2)|^2}.
\end{align}
When $m\rightarrow -\infty$, since
\begin{subequations}
\label{dNLS-fg-xi2-}
\begin{align}
&f\simeq4e^{2\text{Re}(\xi_1-\xi_2)}(|\sinh(k_1+k^*_2)|^2+e^{4\text{Re}\xi_2}|\sinh(k_1-k_2)|^2),\\
&g\simeq-8e^{2\text{Re}\xi_1+2i\text{Im}(k_1+k_2+\xi_2)}\sinh a_2\sinh(k_1-k_2)\sinh(k_1+k^*_2),
\end{align}	
\end{subequations}
the $k_2$-soliton thereby becomes
\begin{align}
\label{dNLS-u-xi2-}
u\simeq-\dfrac{2e^{2\xi_2+2i\text{Im}(k_1+k_2)}\sinh a_2\sinh(k_1-k_2)\sinh(k_1+k^*_2)}
{|\sinh(k_1+k^*_2)|^2+e^{4\text{Re}\xi_2}|\sinh(k_1-k_2)|^2}.
\end{align}
The envelops of $k_2$-soliton \eqref{dNLS-u-xi2+} and \eqref{dNLS-u-xi2-} are
\begin{align}
\label{dNLS-u-xi2-mo}
|u|^2\simeq\dfrac{4\sinh^2a_2(\cosh a^-_{12}-\cos b^-_{12})(\cosh a^+_{12}-\cos b^-_{12})e^{2a_2n}\rho_2^m}
{(\cosh a^{\mp}_{12}-\cos b^-_{12}+(\cosh a^{\pm}_{12}-\cos b^-_{12})e^{2a_2n}\rho_2^m)^2}, \quad m\rightarrow \pm \infty,
\end{align}
which travel with amplitude $\sinh^2a_2$,
velocity $-\ln\rho_2/(2a_2)$, and trajectory
\begin{align}
n(m)=-\frac{1}{2a_2}\left(m\ln\rho_2\mp
\ln\left(\frac{\cosh a^{-}_{12}-\cos b^{-}_{12}}{\cosh a^{+}_{12}-\cos b^{-}_{12}}\right)\right), \quad m\rightarrow \pm \infty.
\end{align}
The phase shift is consequently expressed as
\begin{align}
n(m)=\frac{1}{a_2}
\ln\left(\frac{\cosh a^{-}_{12}-\cos b^{-}_{12}}{\cosh a^{+}_{12}-\cos b^{-}_{12}}\right).
\end{align}
The two-soliton solution $|u|^2$ given by \eqref{dNLS-2ss-fg-xi} is illustrated in Figure 2.
\vskip35pt
\begin{center}
\begin{picture}(120,80)
\put(-120,-30){\resizebox{!}{4.0cm}{\includegraphics{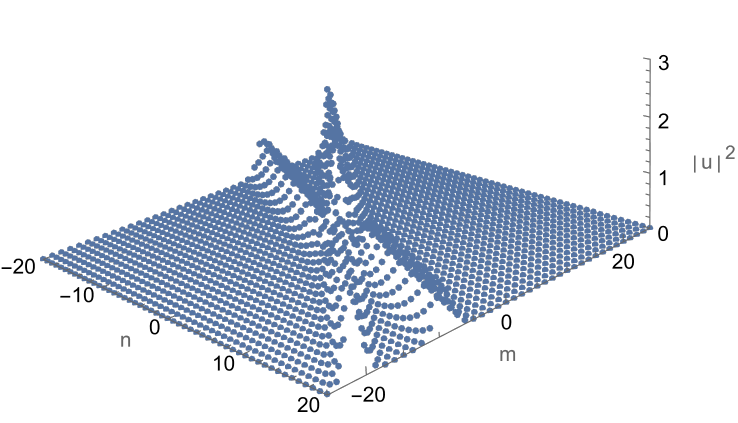}}}
\put(100,-23){\resizebox{!}{4.0cm}{\includegraphics{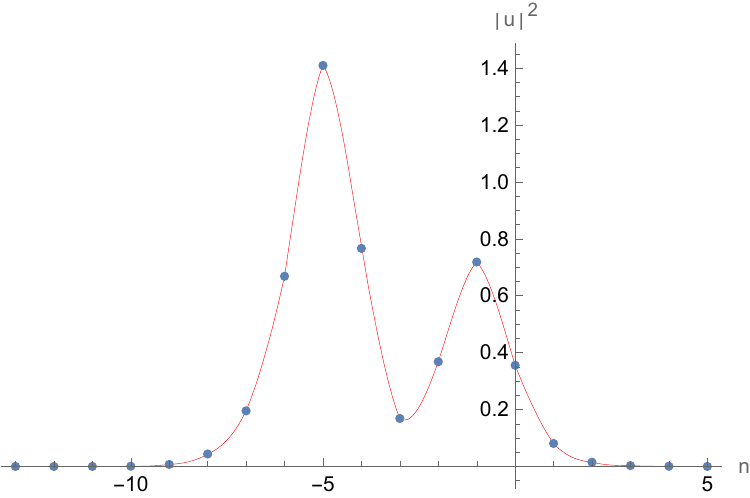}}}
\end{picture}
\end{center}
\vskip 30pt
\begin{center}
\begin{minipage}{15cm}{\footnotesize
\qquad\qquad\qquad\qquad\qquad\quad(a)\qquad\qquad\qquad\qquad\qquad\qquad\qquad\qquad\qquad\qquad\qquad\quad (b) \\
{\bf Figure 2}. Two-soliton solution $|u|^2$ given by \eqref{dNLS-2ss-fg-xi} with $c_1=d_1=c_2=d_2=1$, $\delta=-0.6$,
$k_1=0.4+0.1i$ and $k_2=0.5+0.6i$: (a) shape and movement; (b) 2D-plot of (a) at $m=7$.}
\end{minipage}
\end{center}

\noindent{\bf Remark 2}. When $k_2=ik^*_1$, \eqref{dNLS-2ss-fg-xi} generates the breather solution (see \cite{LWZ-SAPM}), whose
dynamical behaviors are illustrated in Figure 3.
\vskip35pt
\begin{center}
\begin{picture}(120,80)
\put(-150,-30){\resizebox{!}{5.0cm}{\includegraphics{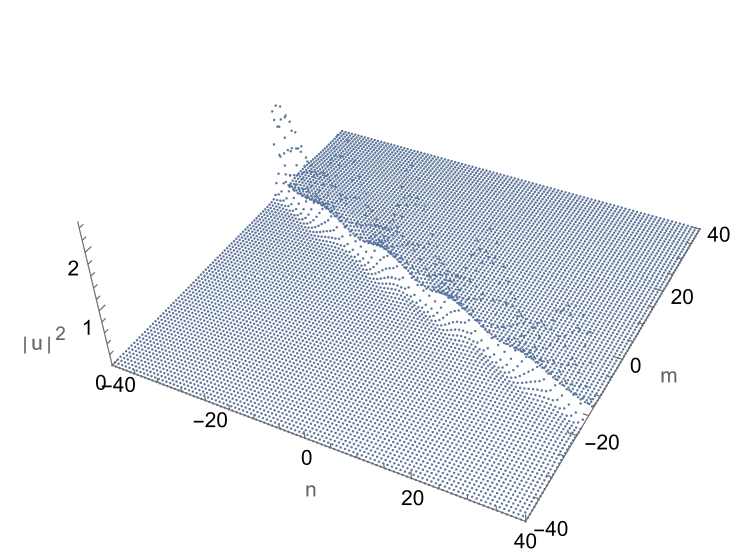}}}
\put(80,-23){\resizebox{!}{4.0cm}{\includegraphics{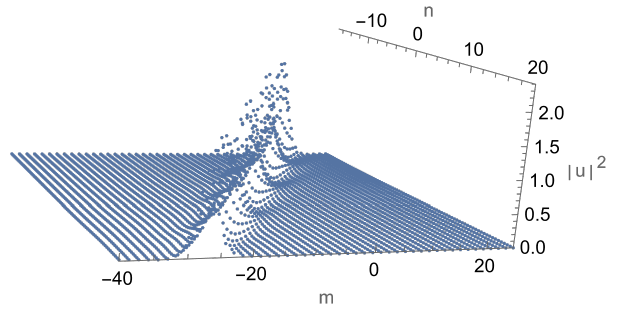}}}
\end{picture}
\end{center}
\vskip 20pt
\begin{center}
\begin{minipage}{15cm}{\footnotesize
\qquad\qquad\qquad\qquad\qquad\quad(a)\qquad\qquad\qquad\qquad\qquad\qquad\qquad\qquad\qquad\qquad\qquad (b) \\
{\bf Figure 3.} Breather solution $|u|^2$ given by \eqref{dNLS-2ss-fg-xi} with $\delta=-0.6$:
(a) $c_1=d_1=c_2=d_2=1$, $k_1=0.5+0.1i$ and $k_2=0.1+0.5i$; (b) $c_1=c_2=1+0.2i$, $d_1=d_2=0.4+0.1i$,
$k_1=0.48+0.08i$ and $k_2=0.08+0.48i$.}
\end{minipage}
\end{center}

\subsubsection{Jordan-block solutions}
Now let $L$ be a Jordan-block matrix
\begin{align}
\label{dNLS-L-Jor}
L=\left(
\begin{array}{cccc}
k_1 &   0 &  \cdots&  0\\
1 & k_1 &  \cdots&  0\\		
\vdots& \ddots&\ddots&\vdots\\
0&\cdots& 1 & k_1
\end{array}\right)_{(N+1)\times(N+1)}.
\end{align}	
Then we immediately obtain the entries in the Casoratian \eqref{dNLS-solu-1}, of form
\begin{align*}
\Phi_{j}=\left\{
\begin{aligned}
&c_1\dfrac{\partial^{j-1}_{k_1}e^{\xi_1}}{(j-1)!}, \quad j=1, 2, \dots, N+1,\\
&d_1\dfrac{\partial^{s-1}_{k_{1}^{*}}e^{-\xi_1^*}}{(s-1)!}, \quad j=N+1+s,\quad s=1,2,\dots,N+1.
\end{aligned}	
\right.
\end{align*}
In the case of $N=1$, the simplest Jordan-block solution is given by \eqref{dNLS-uw-fg},
in which
\begin{subequations}
\label{dNLS-JBS-fg}
\begin{align}
f& =4e^{-4\text{Re}\xi_1}((|d_1|^2+|c_1|^2e^{4\text{Re}\xi_1})^2+4|c_1d_1\eta_1|^2e^{4\text{Re}\xi_1}\sinh^2a_1),\\
g& =16c_1d_1e^{2ib_1}(|c_1|^2e^{2\xi_1}(\cosh a_1-\eta_1^*\sinh a_1) \nn \\
& \qquad+|d_1|^2e^{-2\xi_1^*}(\cosh a_1+\eta_1\sinh a_1))\sinh a_1,
\end{align}
\end{subequations}
where $\eta_1=\partial_{k_1}\xi_1$. This solution is depicted in Figure 4.

\vskip35pt
\begin{center}
\begin{picture}(120,80)
\put(-120,-30){\resizebox{!}{4.5cm}{\includegraphics{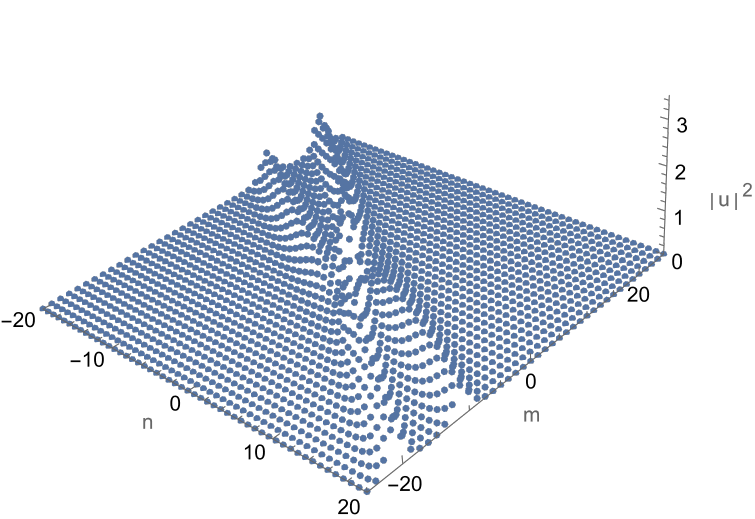}}}
\put(100,-23){\resizebox{!}{4.0cm}{\includegraphics{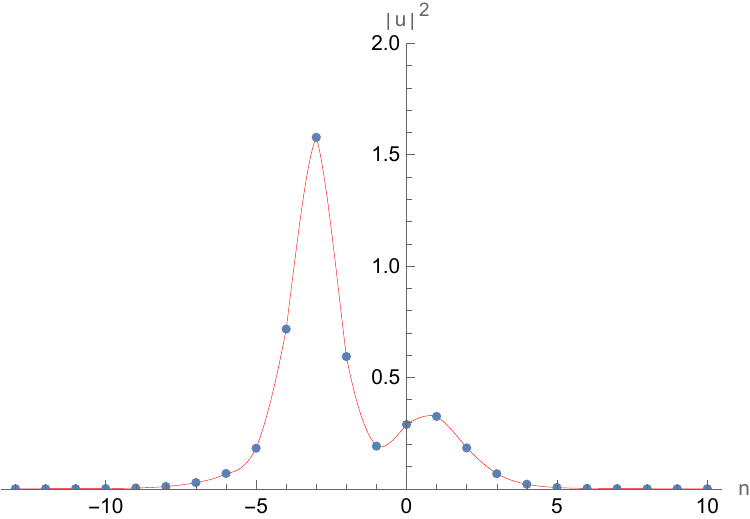}}}
\end{picture}
\end{center}
\vskip 30pt
\begin{center}
\begin{minipage}{15cm}{\footnotesize
\qquad\qquad\qquad\qquad\qquad\quad(a)\qquad\qquad\qquad\qquad\qquad\qquad\qquad\qquad\qquad\qquad\qquad\quad (b) \\
{\bf Figure 4.} Jordan-block solution $|u|^2$ given by \eqref{dNLS-JBS-fg} with $c_1=d_1=1$ and $\delta=-1$:
(a) shape and movement with $k_1=0.4+0.1i$; (b) 2D-plot of (a) at $m=4$.}
\end{minipage}
\end{center}

\subsection{Continuum limits}
\label{sec:dNLS-cl}

The semi-discrete NLS equation and NLS equation, namely the differential-difference NLS
and continuous NLS equations, can be recovered by imposing the continuum limits on the dNLS equation \eqref{dNLS}.
In order to get the semi-discrete and continuous versions of
the dNLS equation \eqref{dNLS}, we will firstly use the indicative formula
\begin{align}
\label{dc-re}
\lim\limits_{m\rightarrow\infty}(1+k/m)^m=e^k,
\end{align}
in the purpose to give the discrete-continuous connections
between the various parameters in the equation and the lattice spacing,
where the discrete exponential functions \eqref{xi}, i.e.,
\begin{align}
\label{xj-def}
e^{\xi_j}=e^{k_jn}
\bigg(\dfrac{i-2\delta e^{k_j}\sinh k_j}{i-2\delta e^{-k_j}\sinh k_j}\bigg)^{\frac{m}{2}}, \quad
j=1,2,\ldots,N+1
\end{align}
should be considered.

\subsubsection{Semi-continuous limit:}
\label{Semi-CL}

To use \eqref{xj-def}
for a straight limit in the $m$-direction, we write
\begin{align}
\label{xj-def-modi}
\bigg(\dfrac{i-2\delta e^{k_j}\sinh k_j}{i-2\delta e^{-k_j}\sinh k_j}\bigg)^{\frac{m}{2}}
=\bigg(1-\dfrac{4\delta\sinh^2k_j}{i-2\delta e^{-k_j}\sinh k_j}\bigg)^{\frac{m}{2}},
\quad j=1,2,\ldots,N+1
\end{align}
and therefore $\delta$ must approach zero as $m\rightarrow \infty$, i.e., $m\delta=z$. The
discrete exponential functions \eqref{xj-def} become
\begin{align}
\label{lambdaj-def}
e^{\xi_j} \rightarrow e^{\lambda_j}, \quad \text{with} \quad \lambda_j:=k_{j}n+2iz\sinh^2k_{j},
\quad j=1,2,\ldots,N+1,
\end{align}
and $u(n,m)\rightarrow u(n,z)$, $w(n,m)\rightarrow w(n,z)$.
Inserting the Taylor series
\begin{subequations}
\label{Tay-exp-dsd}
\begin{align}
& \wh{u}=u(z+\delta)=u+\delta u_z+\ldots, \\
& \wh{\dt{u}}=\dt{u}(z+\delta)=\dt{u}+\delta\dt{u}_z+\ldots,
\end{align}
\end{subequations}
into the equation \eqref{ndNLS} and noting that $w\rightarrow 1$, we know that $u(n,z)$ satisfies the differential-difference equation
\begin{align}
\label{sdNLS-z}
iu_z-2u+(1+|u|^2)(\wt{u}+\dt{u})=0,
\end{align}
where $u_z$ denotes the derivative with respect to $z$. It is readily seen that equation \eqref{sdNLS-z} is exactly the
semi-discrete NLS equation \cite{AL-JMP}.

Double Casoratian solutions of the equation \eqref{sdNLS-z} have
been discussed in \cite{DLZ-AMC}, which are summarized as follows.
\begin{Thm}
The function
\begin{align}
\label{sdNLS-u-fg}
u=g/f
\end{align}
with
\begin{align}
\label{sdNLS-z-solu}
f=|e^{-N\Ga}\wh{\Phi^{(N)}};e^{N\Ga}T\wh{\Phi^{*(-N)}}|,\quad
g=|e^{-N\Ga}\widehat{\Phi^{(N+1)}};e^{N\Ga}T\wh{\Phi^{*(-N+1)}}|
\end{align}
solves the equation \eqref{sdNLS-z}, where $\Phi=e^{\Ga n+2iz\sinh^2 \Ga}C^{+}$
and $T$ is a constant matrix of order $2(N+1)$ satisfying matrix equations
\begin{align}
\label{sdNLS-z-Ga-T}
\Ga T+ T\Ga^*=\bm 0,\quad TT^*=-I.
\end{align}
\end{Thm}

For the semi-discrete NLS equation \eqref{sdNLS-z}, its
one-soliton solution in double Casoratian form was presented explicitly in \cite{DLZ-AMC}, while the two-soliton
solution, as well as the Jordan-block solutions has not been touched yet. In what follows, we are interested in the
dynamics of the multi-soliton solutions and Jordan-block solutions for the equation \eqref{sdNLS-z}. To this end, we still
consider the block matrices $\Ga$ and $T$ in
\eqref{dNLS-LT-form} together with \eqref{L-diag} and \eqref{dNLS-L-Jor}.

\noindent {\it Soliton solutions:} For the diagonal matrix \eqref{L-diag}, equation \eqref{sdNLS-z-solu} with
\begin{align}
\Phi_{j}=\left\{
\begin{aligned}
& c_je^{\lambda_j},\quad j=1, 2, \ldots, N+1,\\
& d_je^{-\lambda_s^*},\quad j=N+1+s, \quad s=1,2,\ldots,N+1.
\end{aligned}	
\right.
\end{align}
supplies with the multi-soliton solutions.

In the case of $N=0$, the one-soliton solution is described as
\begin{align}
\label{sdNLS-1ss-u}
u=\frac{2c_1d_1e^{i(b_1+2\text{Im}\lambda_1)}\sinh a_1}
{|c_1|^2e^{2\text{Re}\lambda_1}+|d_1|^2e^{-2\text{Re}\lambda_1}},
\end{align}
and the carrier wave reads
\begin{align}
\label{sdNLS-1ss-u-mo}
|u|^2=\sinh^2a_1\sech^2(a_1n-2z\sin b_1\sinh a_1+\ln|c_1/d_1|).
\end{align}	
When $b_1=\ell\pi$ with $\ell\in\mathbb{Z}$, i.e., $k_1$ being a real number, then \eqref{sdNLS-1ss-u-mo} exhibits a
stationary wave. While when $b_1\neq\ell\pi$, equation \eqref{sdNLS-1ss-u-mo} describes a unidirectional soliton
traveling with velocity $2(\sin b_1\sinh a_1)/a_1$, amplitude $\sinh^2 a_1$, initial phase $\ln|c_1/d_1|$ and
top trace
\[n(z)=(2z\sin b_1\sinh a_1-\ln|c_1/d_1|)/a_1.\]
We depict the soliton \eqref{sdNLS-1ss-u-mo} in Figure 5.
\vskip35pt
\begin{center}
\begin{picture}(120,80)
\put(-160,-23){\resizebox{!}{3.5cm}{\includegraphics{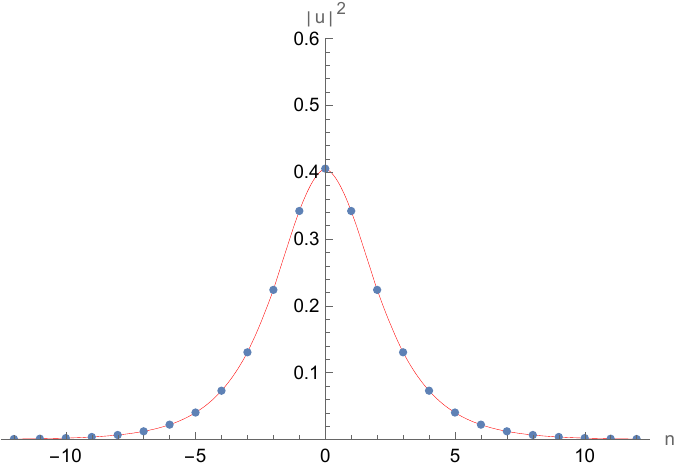}}}
\put(-20,-23){\resizebox{!}{4.0cm}{\includegraphics{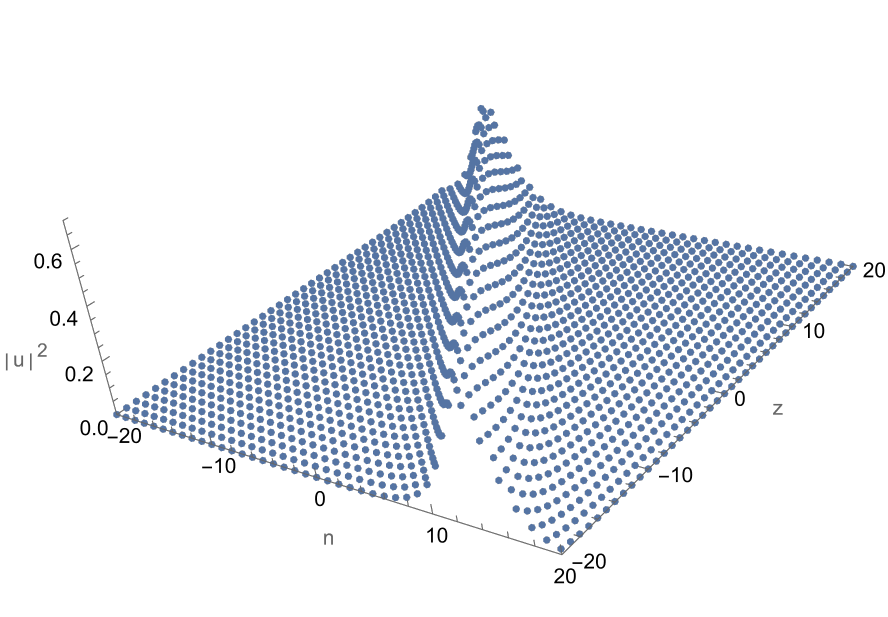}}}
\put(160,-23){\resizebox{!}{3.5cm}{\includegraphics{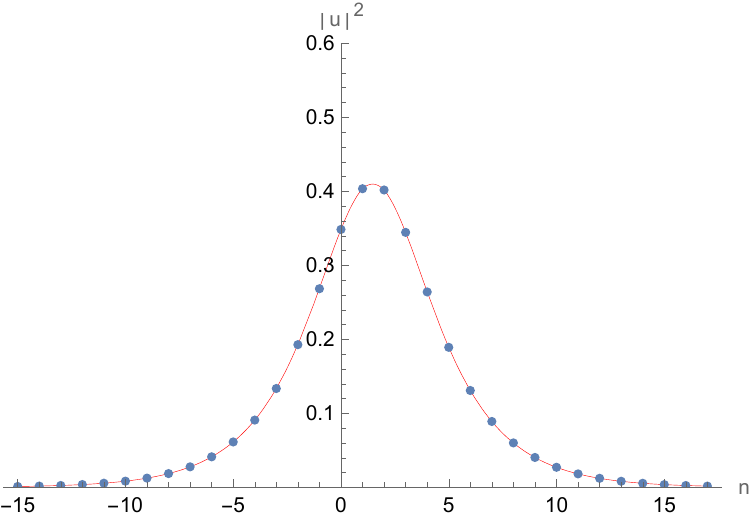}}}
\end{picture}
\end{center}
\vskip 20pt
\begin{center}
\begin{minipage}{15cm}{\footnotesize
\qquad\qquad\qquad(a)\qquad\qquad\qquad\qquad\qquad\qquad\qquad\qquad\quad(b)\qquad\qquad\qquad\qquad\qquad\qquad\qquad\quad (c) \\
{\bf Figure 5}. One-soliton solution $|u|^2$ given by \eqref{sdNLS-1ss-u-mo} with $c_1=d_1=1$:
(a) a stationary wave with $k_1=0.3+\pi i/2$;
(b) shape and movement with $k_1=0.2-0.4i$;
(c) 2D-plot of (b) at $z=-1$.}
\end{minipage}
\end{center}

In the case of  $N=1$, the two-soliton solution is of form \eqref{sdNLS-u-fg},
where $f$ and $g$ are defined by \eqref{dNLS-2ss-fg-xi} with $\xi_j\rightarrow \lambda_j,~(j=1,2)$.
Similar to the discrete case, here we still call the
asymptotic solitons as $k_1$-soliton and $k_2$-soliton, respectively. We denote $\lambda_j$ as $\lambda_j=\mathcal{C}_j+i\mathcal{D}_j,~(j=1,2)$
with
\begin{align}
\label{sdNLS-lambdaj}
\mathcal{C}_j=(a_jn-2z\sinh a_j\sin b_j)/2, \quad \mathcal{D}_j=b_jn/2+(\cosh a_j\cos b_j-1)z,
\end{align}
and assume $c_j=d_j=1$, $0<a_1<a_2$ and $0<b_1<b_2<\pi/2$ to guarantee $\mathcal{Y}=a_2\sinh a_1\sin b_1-a_1\sinh a_2 \sin b_2<0$. By the analogous
analysis above, for fixed $\lambda_1$, we determine the $k_1$-soliton
\begin{align}
\label{sdNLS-lambda-1-u}
u\simeq
\left\{
\begin{aligned}
& \dfrac{2e^{2\lambda_1+2i\text{Im}(k_1+k_2)}\sinh a_1\sinh(k^*_1-k^*_2)
\sinh(k_1+k^*_2)}{|\sinh(k_1-k_2)|^2+|\sinh(k_1+k^*_2)|^2 e^{4\text{Re}\lambda_1}},\quad z\rightarrow +\infty, \\
& \dfrac{2e^{2\lambda_1+2i\text{Im}(k_1+k_2)}\sinh a_1\sinh(k_1-k_2)
\sinh(k^*_1+k_2)}{|\sinh(k_1+k^*_2)|^2+|\sinh(k_1-k_2)|^2 e^{4\text{Re}\lambda_1}},\quad z\rightarrow -\infty,
\end{aligned}	
\right.
\end{align}
and the corresponding envelop
\begin{align}
\label{sdNLS-u-lambda1-mo}
|u|^2\simeq\dfrac{4\sinh^2a_1(\cosh a^-_{12}-\cos b^-_{12})(\cosh a^+_{12}-\cos b^-_{12})e^{-2a_1n+4z\sinh a_1\sin b_1}}
{(\cosh a^{\pm}_{12}-\cos b^-_{12}+(\cosh a^{\mp}_{12}-\cos b^-_{12})e^{-2a_1n+4z\sinh a_1\sin b_1})^2}, \quad z\rightarrow \pm\infty.
\end{align}
For fixed $\lambda_2$, the $k_2$-soliton
\begin{align}
\label{sdNLS-lambda2-u}
u\simeq
\left\{
\begin{aligned}
& -\dfrac{2e^{2\lambda_2+2i\text{Im}(k_1+k_2)}\sinh a_2 \sinh(k_1-k_2)\sinh(k_1+k^*_2)}
{|\sinh(k_1+k^*_2)|^2+e^{4\text{Re}\lambda_2}|\sinh(k_1-k_2)|^2}, \quad z\rightarrow +\infty, \\
& -\dfrac{2e^{2\lambda_2+2i\text{Im}(k_1+k_2)}\sinh a_2 \sinh(k^*_1-k^*_2)\sinh(k^*_1+k_2)}
{|\sinh(k_1-k_2)|^2+e^{4\text{Re}\lambda_2}|\sinh(k_1+k^*_2)|^2}, \quad z\rightarrow -\infty,
\end{aligned}	
\right.
\end{align}
and the wave package
\begin{align}
\label{sdNLS-lambda2-u-mo}
|u|^2\simeq\dfrac{4\sinh^2a_2(\cosh a^-_{12}-\cos b^-_{12})(\cosh a^+_{12}-\cos b^-_{12})e^{-2a_2n+4z\sinh a_2\sin b_2}}
{(\cosh a^{\mp}_{12}-\cos b^-_{12}+(\cosh a^{\pm}_{12}-\cos b^-_{12})e^{-2a_2n+4z\sinh a_2 \sin b_2})^2}, \quad z\rightarrow \pm\infty,
\end{align}
can be derived out. The $k_j$-soliton travels with amplitude $\sinh^2a_j$,
velocity $(2\sinh a_j\sin b_j)/a_j$, and trajectory
\begin{align}
n(z)=\frac{1}{2a_j}\left(4z\sinh a_j\sin b_j \mp (-1)^{j}
\ln\left(\frac{\cosh a^{-}_{12}-\cos b^{-}_{12}}{\cosh a^{+}_{12}-\cos b^{-}_{12}}\right)\right), \quad z\rightarrow \pm \infty.
\end{align}
The phase shift is described as
\begin{align}
\frac{\mp (-1)^{j}}{a_j}\ln\left(\frac{\cosh a^{-}_{12}-\cos b^{-}_{12}}{\cosh a^{+}_{12}-\cos b^{-}_{12}}\right).
\end{align}
We depict the this two-solution solution in Figure 6.
\vskip35pt
\begin{center}
\begin{picture}(120,80)
\put(-120,-23){\resizebox{!}{5.0cm}{\includegraphics{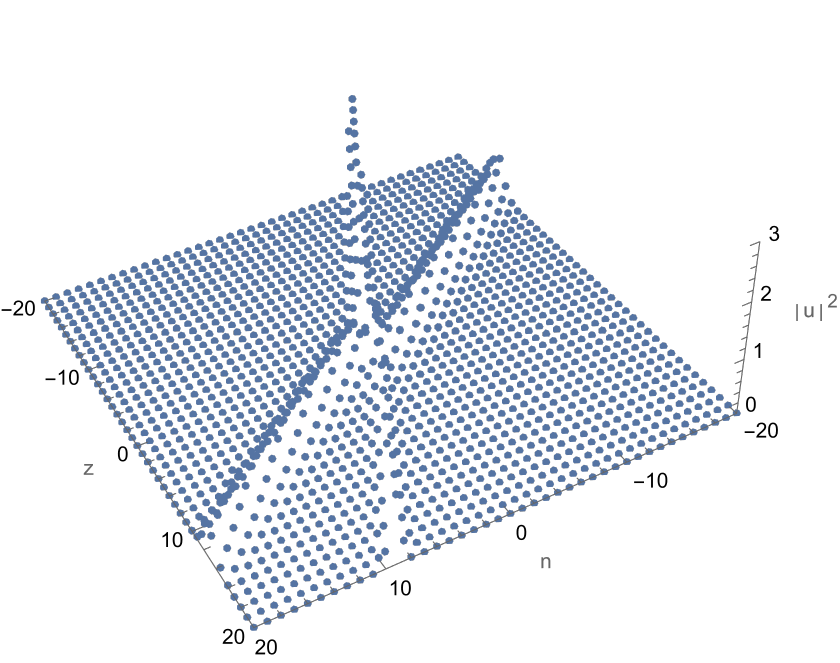}}}
\put(100,-23){\resizebox{!}{4cm}{\includegraphics{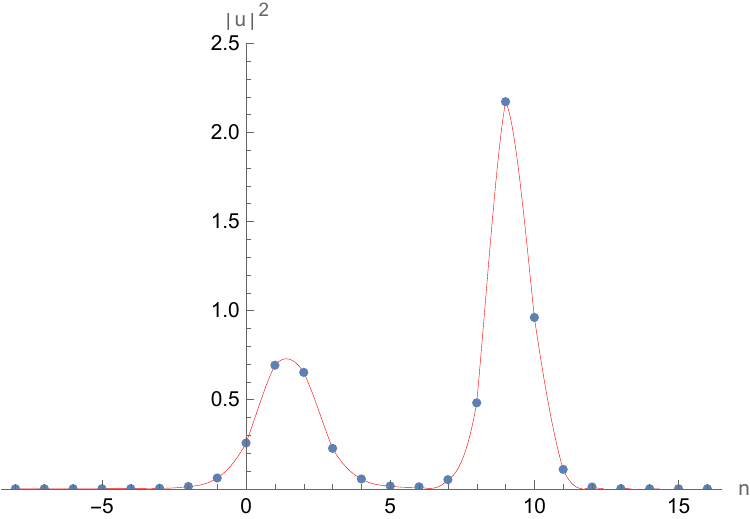}}}
\end{picture}
\end{center}
\vskip 20pt
\begin{center}
\begin{minipage}{15cm}{\footnotesize
\qquad\qquad\qquad\qquad\qquad\qquad\quad(a)\qquad\qquad\qquad\qquad\qquad\qquad\qquad\qquad\qquad\qquad (b) \\
{\bf Figure 6.} Two-soliton solution $|u|^2$ for the equation \eqref{sdNLS-z} with $k_1=0.4+1.4i$,
$k_2=0.6+1.1i$ and $c_1=d_1=c_2=d_2=1$:
(a) shape and movement; (b) 2D-plot of (a) at $z=4$.}
\end{minipage}
\end{center}

\noindent {\it Jordan-block solutions:}
To proceed, let $L$ be the Jordan-block matrix \eqref{dNLS-L-Jor}, the basic entries $\{\Phi_j\}$ are thereby of the form
\begin{align}
\label{sdNLS-JBS}
\Phi_{j}=\left\{
\begin{aligned}
&c_1\dfrac{\partial^{j-1}_{k_1}e^{\lambda_1}}{(j-1)!},\quad j=1, 2, \ldots, N+1,\\
&d_1\dfrac{\partial^{s-1}_{k_1^*}e^{-\lambda_1^*}}{(s-1)!}, \quad j=N+1+s, \quad s=1,2,\ldots,N+1.
\end{aligned}	
\right.
\end{align}
In this case, the simplest Jordan-block solution is still given by \eqref{sdNLS-u-fg}
and \eqref{dNLS-JBS-fg} up to a replacement of $\xi_1$ by $\lambda_1$.
The behavior of $|u|^2$ is depicted in Figure 7.
\vskip35pt
\begin{center}
\begin{picture}(120,80)
\put(-120,-23){\resizebox{!}{5.0cm}{\includegraphics{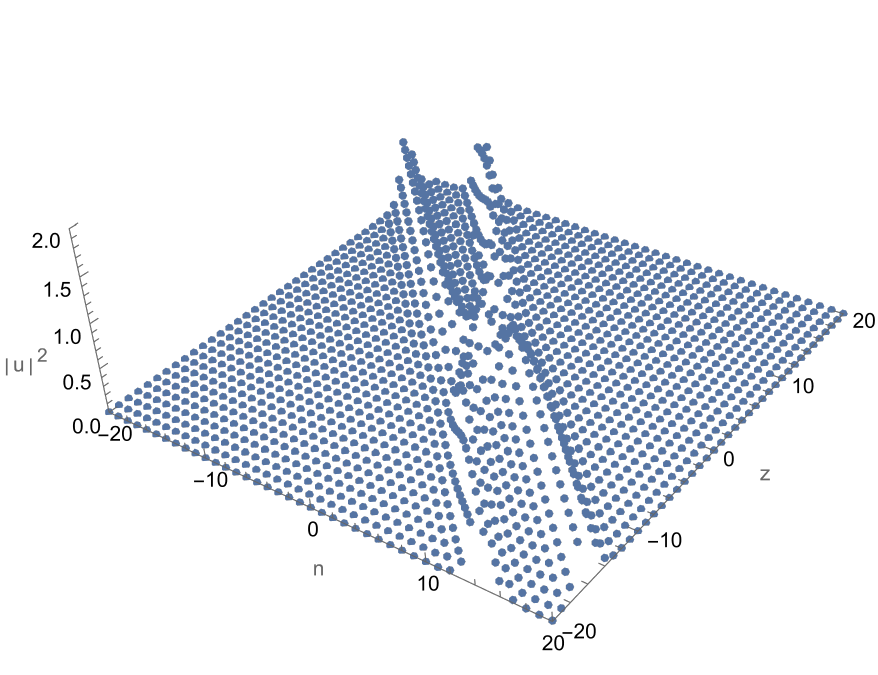}}}
\put(100,-23){\resizebox{!}{4.0cm}{\includegraphics{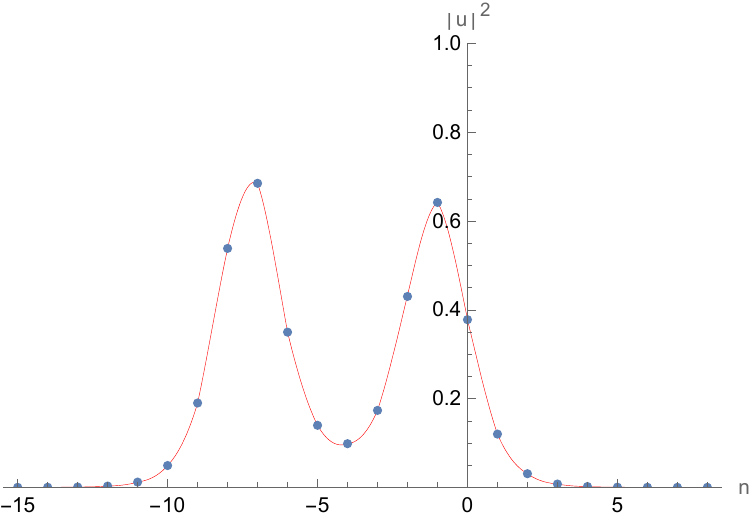}}}
\end{picture}
\end{center}
\vskip 20pt
\begin{center}
\begin{minipage}{15cm}{\footnotesize
\qquad\qquad\qquad\qquad\qquad\quad(a)\qquad\qquad\qquad\qquad\qquad\qquad\qquad\qquad\qquad\qquad\qquad\quad (b) \\
{\bf Figure 7.} Jordan-block solution $|u|^2$ given by \eqref{sdNLS-JBS} with $k_1=0.4-0.25i$ and $c_1=d_1=1$:
(a) shape and movement; (b) 2D-plot of (a) at $z=4$.}
\end{minipage}
\end{center}

\subsubsection{Full-continuous limit:} Now, we change the parameters $l_j=k_j/\epsilon$, $t=\epsilon^2z$
and perform the limit on the discrete variable $n$,
\begin{align}
\label{sdNLS-FCM-n}
n \rightarrow \infty, \quad \epsilon \rightarrow 0, \quad \text{s.t.} \quad x=n\epsilon \quad \text{fixed}.
\end{align}	
The plane wave functions \eqref{lambdaj-def} take the form of
\begin{align}
e^{\lambda_j} \rightarrow e^{l_jx+2il_j^2t}, \quad j=1,2,\ldots,N+1.
\end{align}	
Reinterpreting the variable $u$ as $u(n,z)\rightarrow\epsilon u(x,t)/\sqrt{2} $, we arrive at the NLS
equation \eqref{NLS}, whose solution is described as follows.
\begin{Thm}
The function \eqref{sdNLS-u-fg} with\footnote{For the basic
column vector $\Phi:=\Phi(x,t)$, $\wh{\Phi^{(N)}}$ is the sequence $\Phi^{(0)}, \Phi^{(1)},\ldots,\Phi^{(N)}$,
where $\Phi^{(j)}=\frac{\partial^j\Phi}{\partial x^j}$.}
\begin{align}
\label{NLS-solu}
f=|e^{-N\Ga}\wh{\Phi^{(N)}};e^{N\Ga}T\wh{\Phi^{*(-N)}}|,\quad
g=2\sqrt{2}|e^{-N\Ga}\widehat{\Phi^{(N+1)}};e^{N\Ga}T\wh{\Phi^{*(-N+1)}}|,
\end{align}
solves the equation \eqref{NLS}, where $\Phi=e^{\Ga x+2i\Ga^2t}C^{+}$
and $T$ is a constant matrix of order $2(N+1)$ satisfying matrix equations
\begin{align}
\label{cNLS-Ga-T}
\Ga T+ T\Ga^*=\bm 0,\quad TT^*=-I.
\end{align}
\end{Thm}

\section{ndNLS equation \eqref{ndNLS}: solutions and continuum limits}

In this part, we will discuss solutions and continuum limits of the
ndNLS equation \eqref{ndNLS}, which is reduced from the dAKNS(2) equation \eqref{dAKNS2}
by imposing reduction $(v,w)=(u_{-n,-m},\wh{w}_{-n,-m})$.
In what follows, we appoint that
for a function $\hbar=\hbar_{x_1,x_2}$, notation $\hbar_{-1}$ means the reverse space-time of $\hbar$, i.e.,
$\hbar_{-1}=\hbar_{-x_1,-x_2}$. In the discrete case, $\wt{\hbar}_{-1}=\hbar_{-n-1,-m}$,
$\wh{\hbar}_{-1}=\hbar_{-n,-m-1}$ and $\wh{\dt{\hbar}}_{-1}=\hbar_{-n+1,-m-1}$.

\subsection{Solutions to equation \eqref{ndNLS}}

Following the treatment in the above section, we still take
$M=N$ and replace $C^{\pm}$ by $e^{\mp N\Ga}C^{\pm}$ in \eqref{dAKNS2-dCs}.
We then have the following result.
\begin{Thm}
\label{Thm-ndNLS-solu}
The functions \eqref{dNLS-uw-fg} with
\begin{align}
\label{ndNLS-solu}
f=|e^{-N\Ga}\widehat{\Phi^{(N)}};e^{N\Ga}T\widehat{\Phi_{-1}^{(-N)}}|,\quad
g=|e^{-N\Ga}\widehat{\Phi^{(N+1)}};e^{N\Ga}T\widehat{\Phi_{-1}^{(-N+1)}}|,
\end{align}
solve the equation \eqref{ndNLS}, where the $2(N+1)$-th order column vector $\Phi$ is defined by
\eqref{dAKNS2-Phi-Ga} and $T\in\mathbb{C}^{(2N+2)\times(2N+2)}$ is a constant matrix satisfying
 matrix equations
\begin{align}
\label{ndNLS-AT}
\Ga T-T\Ga=\bm 0,\quad T^{2}=|e^{\Ga}|^{2}I.
\end{align}
\end{Thm}

\noindent{\bf Remark 3.} {\it The matrix equations listed in \eqref{ndNLS-AT} are same as the constraints given in
\cite{DLZ-AMC} (Eq. (30)) to discuss double Casoratian solutions for a nonlocal semi-discrete NLS equation. There are three choices of the
matrix $\Ga=\text{Diag}(L,\Lambda)$ with $\Lambda=-L$ or $\Lambda=-L^*$ or $\Lambda=L^*$. For simplicity, we just
consider the first case to show the soliton solutions and Jordan-block solutions of the ndNLS equation \eqref{ndNLS}.}

We take $\Ga$ and $T$ as
\begin{align}
\label{ndNLS-Ga-T-ex}
\Ga=\left(
\begin{array}{cc}
L & \bm 0  \\
\bm 0 & -L
\end{array}\right),\quad
T=\left(
\begin{array}{cc}
I & \bm 0  \\
\bm 0 & -I
\end{array}\right),
\end{align}	
where $L$ is a Jordan canonical matrix.
Due to the block structure of matrix $T$, we note that $C^{+}$ can be
gauged to be $\bar{I}=(1,1,\dots,1;1,1,\dots,1)^{\st}$ or $\breve{I}=(1,0,\dots,0;1,0,\dots,0)^{\st}$.
It implies that the solutions of equation \eqref{ndNLS} are independent of phase parameter $C^{+}$,
i.e., the initial phase has always to be $0$. Besides the notation $\xi_j$ introduced in \eqref{xi}, we need another one
\begin{align}
\label{theta}
\theta_s=-k_sn+\iota_s m, \quad e^{\iota_s}=\bigg(
\dfrac{i+2\delta e^{-k_s}\sinh k_s}{i+2\delta e^{k_s}\sinh k_s}\bigg)^{\frac{1}{2}}, \quad s=1,2,\ldots,N+1.
\end{align}	

\subsubsection{Soliton solutions}
With the diagonal matrix $L$ \eqref{L-diag} and $C^{+}=\bar{I}$, $\Phi$ is thereby formulated as
\begin{align}
\label{ndNLS-Phij}
\Phi_{j}=\left\{
\begin{aligned}
& e^{\xi_j},\quad j=1, 2, \ldots, N+1, \\
& e^{\theta_s},\quad j=N+1+s,\quad s=1,2,\ldots,N+1.
\end{aligned}	
\right.
\end{align}

In the case of $N=0$, we have one-soliton solution
\begin{subequations}
\label{ndNLS-1ss}
\begin{align}
\label{ndNLS-1ss-u}
& u=e^{\xi_1+\theta_1}\sinh (2k_{1})\sech (\xi_1-\theta_1), \\
& w=\wh{\dt{f}}f/(\wh{f}\dt{f}), \quad f=-2\cosh(\xi_1-\theta_1).
\end{align}
\end{subequations}
To clarify one-soliton dynamics, we still take $k_1=(a_1+ib_1)/2$ and
introduce the following variables
\begin{subequations}
\label{ndNLS-nota}
\begin{align}
\label{Xa3}
& \d{X}_{a_j}:=\d{X}_{a_j}^{b_j}=1+\delta e^{a_j}\sin b_j,\quad \bar{Y}_{a_j}=-(\bar{X}_{a_j})^2-\cd{X}_{a_j}\d{X}_{a_j},\quad
\cd{Y}_{a_j}=\bar{X}_{a_j}(\cd{X}_{a_j}-\d{X}_{a_j}), \\
& Z_1=\bar{Y}_{a_1}\bar{Y}_{-a_1}-\cd{Y}_{a_1}\cd{Y}_{-a_1},\quad Z_2=\bar{Y}_{a_1}\cd{Y}_{-a_1}+\cd{Y}_{a_1}\bar{Y}_{-a_1}, \\
& \alpha=\arcsin(2Z_1Z_2/(Z_1^2+Z_2^2)),\quad \beta_{a_1}=e^{a_1n}((\bar{Y}_{-a_1})^2+(\cd{Y}_{-a_1})^2)^{-\frac{m}{2}},\\
& \beta_{-a_1}=\beta_{a_1}|_{a_1\rightarrow -a_1},\quad
\gamma=\frac{((\bar{X}_{a_1})^2+(\cd{X}_{a_1})^2)((\bar{X}_{-a_1})^2+(\d{X}_{-a_1})^2)}{((\bar{X}_{a_1})^2
+(\d{X}_{a_1})^2)((\bar{X}_{-a_1})^2+(\cd{X}_{-a_1})^2)},
\end{align}
\end{subequations}
where $\bar{X}_{a_1}$ and $\cd{X}_{a_1}$ are defined by \eqref{dNLS-Xa12} and
$\bar{Y}_{-a_1},\cd{Y}_{-a_1},\d{X}_{-a_1}$ are the results determined by replacing $a_1$ in
\eqref{Xa3} by $-a_1$. The corresponding envelop is
\begin{align}
\label{ndNLS-1ss-mo}
|u|^2=\frac{2(\cosh2a_1-\cos2b_1)(\gamma/(Z_1^2+Z_2^2))^{\frac{m}{2}}}
{\beta_{a_1}^2+\beta_{-a_1}^2+2\beta_{a_1}\beta_{-a_1}\cos(2b_1n-\alpha m/2)}.
\end{align}
When $a_1=0$ and $\beta_{a_1}=\beta_{-a_1}$, wave \eqref{ndNLS-1ss-mo} has singularities along points
\begin{align}
\label{ndNLS-1ss-velo}
n(m)=(2l+1)\pi/(2b_1)+\alpha m/(4b_1), \quad l \in \mathbb{Z}.
\end{align}	
In this situation, the carrier wave \eqref{ndNLS-1ss-mo} is oscillatory due to the cosine in the denominator.
We depict such a wave in Figure 8.
\vskip20pt
\begin{center}
\begin{picture}(120,80)
\put(-120,-23){\resizebox{!}{4.0cm}{\includegraphics{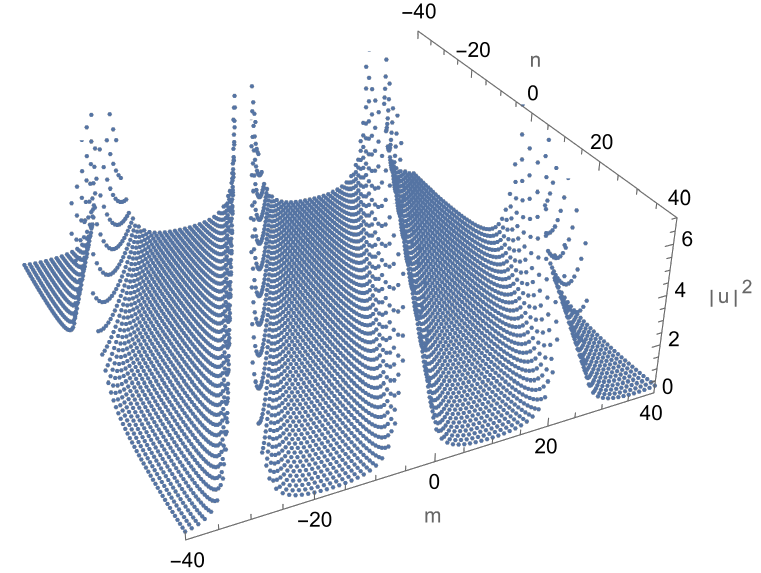}}}
\put(100,-23){\resizebox{!}{3.5cm}{\includegraphics{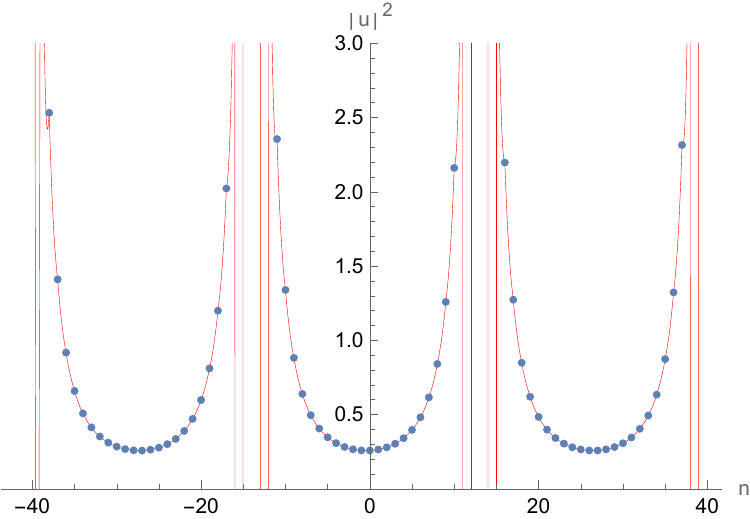}}}
\end{picture}
\end{center}
\vskip 20pt
\begin{center}
\begin{minipage}{15cm}{\footnotesize
\qquad\qquad\qquad\qquad\qquad\qquad\quad(a)\quad\qquad\qquad\qquad\qquad\qquad\qquad\qquad\quad\qquad\qquad(b) \\
{\bf Figure 8}. One-soliton solution $|u|^2$ given by \eqref{ndNLS-1ss-mo} with $\delta=1$ and $k_1=3.2i$:
(a) shape and movement; (b) 2D-plot of (a) at $m=1$.}
\end{minipage}
\end{center}
When $a_1\neq0$, it is evident that wave \eqref{ndNLS-1ss-mo} without singularity reaches its extrema along points $(n,m)$ with
\begin{align}
\label{ndNLS-1ss-a-neq0}
a_1(\beta_{a_1}^2-\beta_{-a_1}^2)-2b_1\beta_{a_1}\beta_{-a_1}\sin(2b_1n-\alpha m/2)=0.
\end{align}	
Then the velocity of \eqref{ndNLS-1ss-mo} is $n'(m)=X/Y$ associated with
\begin{subequations}
\begin{align}
X=& -b_1\beta_{a_1}\beta_{-a_1}\big(\sin(2b_1n-\alpha m/2)
\ln(\beta_{a_1}\beta_{-a_1})/m+\alpha\cos(2b_1n-\alpha m/2)\big) \nn \\
&+2a_1\big(\beta_{a_1}^2(\ln \beta_{a_1}-a_1n)-\beta_{-a_1}^2(\ln \beta_{-a_1}+a_1n)\big)/m,\\
Y=&2a_1^2(\beta_{a_1}^2+\beta_{-a_1}^2)-4b_1^2\beta_{a_1}\beta_{-a_1}\cos(2b_1n-\alpha m/2).
\end{align}	
\end{subequations}
When $a_1$ is near to zero, term $e^{2a_1n}$ or $e^{-2a_1n}$ does not play dominated role and
oscillatory phenomenon would be arisen since the cosine function in the denominator of \eqref{ndNLS-1ss-mo} (see Figure 9(a)).
When $a_1$ is not too small, the effect brought by the cosine function can be neglected and
the soliton appears as the usual bell form (see Figures 9(b) and 9(c)).
\vskip35pt
\begin{center}
\begin{picture}(120,80)
\put(-160,-23){\resizebox{!}{3.5cm}{\includegraphics{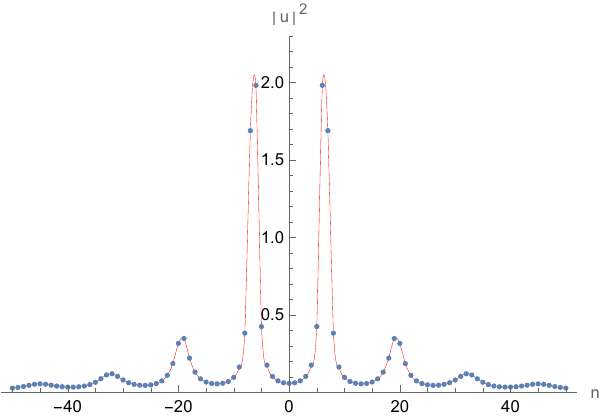}}}
\put(-30,-23){\resizebox{!}{4.0cm}{\includegraphics{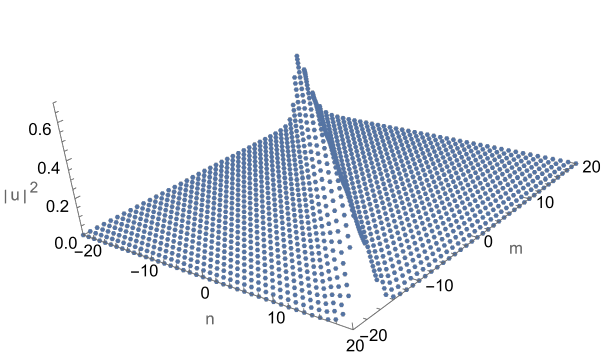}}}
\put(160,-23){\resizebox{!}{3.5cm}{\includegraphics{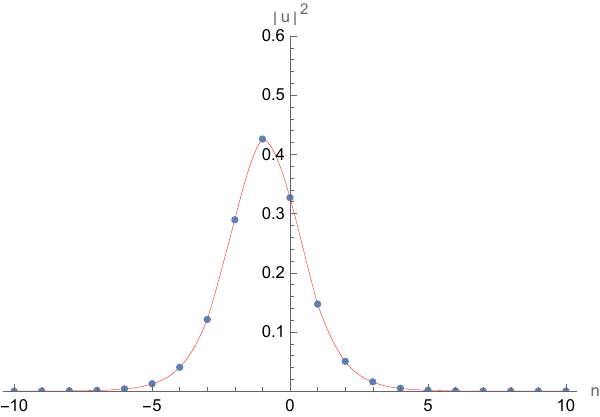}}}
\end{picture}
\end{center}
\vskip 20pt
\begin{center}
\begin{minipage}{15cm}{\footnotesize
\qquad\qquad\qquad(a)\qquad\qquad\qquad\qquad\qquad\qquad\qquad\qquad\quad(b)\qquad\qquad\qquad\qquad\qquad\qquad\qquad\quad (c) \\
{\bf Figure 9.} One-soliton solution $|u|^2$ given by \eqref{ndNLS-1ss-mo} with $\delta=-1.5$:
(a) Oscillatory soliton solution with $k_1=0.01+1.45i$ at $m=0$;
(b) Bell soliton solution with $k_1=0.3+1.5i$; (c) 2D-plot of (b) at $m=1$.}
\end{minipage}
\end{center}

In the case of $N=1$, we get two-soliton solution \eqref{dNLS-uw-fg} with
\begin{subequations}
\label{ndNLS-fg-2ss}
\begin{align}
f& =8(\cosh2k_1\cosh2k_2-1)\cosh(\xi_1-\theta_1)\cosh(\xi_2-\theta_2)-4\sinh2k_1\sinh2k_2\notag\\
&\quad \cdot[\cosh(\xi_1+\xi_2-\theta_1-\theta_2)+2\cosh(\xi_1-\xi_2+\theta_1-\theta_2)-\cosh(\xi_1-\xi_2-\theta_1+\theta_2)],\\
g& =4e^{-\xi_1-\xi_2-\theta_1-\theta_2}(\cosh2k_1-\cosh2k_2)\notag\\
&\quad \cdot[e^{2(\xi_1+\theta_1)}(e^{2\xi_2}+e^{2\theta_2})\sinh2k_1-e^{2(\xi_2+\theta_2)}(e^{2\xi_1}+e^{2\theta_1})\sinh2k_2].
\end{align}
\end{subequations}

We here present the interaction of two-soliton. Besides the decomposition of $\xi_j$ introduced in \eqref{dNLS-xij},
we rewrite $\theta_j,~(j=1,2)$ as
\begin{align}
\label{thetaj}
\theta_j=\mathcal{P}_j+i\mathcal{Q}_j, \quad \text{with} \quad
\mathcal{P}_j=(-2a_jn+m\ln\varsigma_j)/4 \quad \text{and} \quad \mathcal{Q}_j=(-b_jn+m\arg e^{2\iota_j})/2,
\end{align}
in which
\begin{subequations}
\begin{align}
\label{varsigma}
& \varsigma_j=\frac{\delta^2{(e^{-a_j}\cos{b_j}-1)}^2+{(\delta e^{-a_j}\sin b_j+1)}^2}
{\delta^2{(e^{a_j}\cos{b_j}-1)}^2+{(\delta e^{a_j}\sin{b_j}+1)}^2}, \\
& \arg e^{2\iota_j}=\arctan\frac{\bar{X}_{a_j}\d{X}_{-a_j}+\bar{X}_{-a_j}\d{X}_{a_j}}{\bar{X}_{a_j}\bar{X}_{-a_j}-\d{X}_{a_j}\d{X}_{-a_j}},
\end{align}
\end{subequations}
where $\bar{X}_{a_j}$ and $\d{X}_{a_j}$ are defined by \eqref{dNLS-Xa12} and \eqref{Xa3}, respectively.
We find that $e^{\theta_j}\rightarrow \infty$ as
$\mathcal{P}_j\rightarrow +\infty$ and $e^{\theta_j}\rightarrow 0$ as
$\mathcal{P}_j\rightarrow -\infty$. Without loss of generality, we assume $0<a_1<a_2$ and $0<b_1<b_2<\pi/2$ to
guarantee $\mathcal{S}=a_2\ln\varsigma_1-a_1\ln\varsigma_2>0$. For fixed $\theta_1$, from \eqref{thetaj} we deduce that
\begin{align}
\label{theta2-1}
\theta_2=a_2\theta_1/a_1-m\mathcal{S}/(4a_1)+i(a_1\mathcal{Q}_2-a_2\mathcal{Q}_1)/a_1,
\end{align}
which satisfies $e^{\theta_2}\rightarrow 0$ as
$m\rightarrow +\infty$ and $e^{\theta_2}\rightarrow \infty$ as $m\rightarrow -\infty$.
Now, fixing $\xi_1$ again and since
\begin{align}
\label{xij-thetaj}
\xi_j+\theta_j=\ln(\rho_j\varsigma_j)m/4+i(\mathcal{B}_j+\mathcal{Q}_j),
\end{align}
then we obtain
\begin{align}
\label{xi2-theta2}
\xi_2+\theta_2=[(\xi_1+\theta_1)\ln(\rho_2\varsigma_2)+i(\ln(\rho_1\varsigma_1)
(\mathcal{B}_2+\mathcal{Q}_2)-\ln(\rho_2\varsigma_2)(\mathcal{B}_1+\mathcal{Q}_1))]/\ln(\rho_1\varsigma_1),
\end{align}
from which we recognize that the modulus of $e^{\xi_2+\theta_2}$ is fixed, but the argument varies by
the independent variable $m$. After neglecting subdominant exponential terms, it follows from equations \eqref{ndNLS-fg-2ss} that
\begin{subequations}
\label{xi2-theta2-1-fg}
\begin{align}
&f\simeq4e^{-\xi_1\pm\xi_2-\theta_1\mp\theta_2}[e^{2\xi_1}\sinh^2(k_1\mp k_2)+e^{2\theta_1}\sinh^2(k_1\pm k_2)], \quad m \rightarrow \pm \infty,\\
&g\simeq4e^{\xi_1\pm \xi_2+\theta_1\mp\theta_2}(\cosh2k_1-\cosh2k_2)\sinh2k_1, \quad m \rightarrow \pm \infty,
\end{align}	
\end{subequations}
and hence the $k_1$-soliton appears
\begin{align}
\label{xi2-theta2-1-u}
u\simeq\dfrac{e^{2(\xi_1+\theta_1)}(\cosh2k_1-\cosh2k_2)\sinh2k_1}
{e^{2\xi_1}\sinh^2(k_1\mp k_2)+e^{2\theta_1}\sinh^2(k_1\pm k_2)}, \quad m \rightarrow \pm \infty,
\end{align}
the corresponding envelop reads
\begin{align}
\label{xi2-theta2-mo-1-u}
 |u|^2\simeq G_1/F^{\pm}_1, \quad m \rightarrow \pm \infty,
\end{align}
where
\begin{subequations}
\label{G1-F1}
\begin{align}
G_1=& (\rho_1\varsigma_1)^m(\cosh2a_1-\cosh2b_1)(\cosh2a_1+\cosh2a_2+\cos2b_1+\cos2b_2 \nn\\
&-4\cosh a_1\cosh a_2\cos b_1\cos b_2-4\sinh a_1\sinh a_2\sin b_1\sin b_2), \quad m \rightarrow \pm \infty,\\
F^{\pm}_1=&e^{2a_1n}\rho^m_1(\cosh a^{\mp}_{12}-\cos b^{\mp}_{12})^2+e^{-2a_1n}\varsigma^m_1(\cosh a^{\pm}_{12}-\cos b^{\pm}_{12})^2 \nn \\
&+2(\rho_1\varsigma_1)^\frac{m}{2}\big[\cos2(\mathcal{Q}_1-\mathcal{B}_1)\big((\cosh a_1\cos b_2-\cosh a_2\cos b_1)^2 \nn \\
&-(\sinh a_1\sin b_2-\sinh a_2\sin b_1)^2\big)-2\sin2(\mathcal{Q}_1-\mathcal{B}_1) \nn \\
&(\sinh a_1\sin b_2-\sinh a_2\sin b_1)(\cosh a_1\cos b_2-\cosh a_2\cos b_1)\big], \quad m \rightarrow \pm \infty.
\end{align}	
\end{subequations}
To continue, keeping $\xi_2$ and $\theta_2$ fixed, then by the relation
\begin{align}
\label{xi1-theta1}
\xi_1+\theta_1=[(\xi_2+\theta_2)\ln(\rho_1\varsigma_1)+i(\ln(\rho_2\varsigma_2)
(\mathcal{B}_1+\mathcal{Q}_1)-\ln(\rho_1\varsigma_1)(\mathcal{B}_2+\mathcal{Q}_2))]/\ln(\rho_2\varsigma_2),
\end{align}
we obtain that as $m \rightarrow \pm \infty$, functions $f$ and $g$ behave asymptotically
\begin{subequations}
\label{xi1-theta1-1-fg}
\begin{align}
&f\simeq4e^{\mp\xi_1-\xi_2\pm\theta_1-\theta_2}
(e^{2\xi_2}\sinh^2(k_1\pm k_2)+e^{2\theta_2}\sinh^2(k_1\mp k_2)), \quad m \rightarrow \pm \infty,\\
&g\simeq-4e^{\mp\xi_1+\xi_2\pm\theta_1+\theta_2}(\cosh2k_1-\cosh2k_2)\sinh2k_2, \quad m \rightarrow \pm \infty,
\end{align}	
\end{subequations}
and the $k_2$-soliton appears
\begin{align}
\label{xi1-theta1-1-u}
u\simeq-\dfrac{e^{2(\xi_2+\theta_2)}(\cosh2k_1-\cosh2k_2)\sinh2k_2}
{e^{2\xi_2}\sinh^2(k_1\pm k_2)+e^{2\theta_2}\sinh^2(k_1\mp k_2)}.
\end{align}
Thus the corresponding envelop reads
\begin{align}
\label{xi1-theta1-mo-1-u}
 |u|^2\simeq G_2/F^{\pm}_2, \quad m \rightarrow \pm \infty,
\end{align}
where $G_2=G_1|_{1\leftrightarrow 2}$ and
\begin{align}
\label{F2}
F^{\pm}_2=&e^{2a_2n}\rho^m_2(\cosh(a_1\pm a_2)-\cos(b_1\pm b_2))^2+e^{-2a_2n}\varsigma^m_2(\cosh(a_1\mp  a_2)-\cos(b_1\mp b_2))^2 \nn \\
&+2(\rho_2\varsigma_2)^\frac{m}{2}[\cos2(\mathcal{Q}_2-\mathcal{B}_2)\big((\cosh a_1\cos b_2-\cosh a_2\cos b_1)^2 \nn \\
&-(\sinh a_1\sin b_2-\sinh a_2\sin b_1)^2\big)-2\sin2(\mathcal{Q}_2-\mathcal{B}_2) \nn \\
&(\sinh a_1\sin b_2-\sinh a_2\sin b_1)(\cosh a_1\cos b_2-\cosh a_2\cos b_1)].
\end{align}	
The solution $|u|^2$ given by \eqref{ndNLS-fg-2ss} is illustrated in Figure 10.
\vskip35pt
\begin{center}
\begin{picture}(120,80)
\put(-160,-23){\resizebox{!}{4.0cm}{\includegraphics{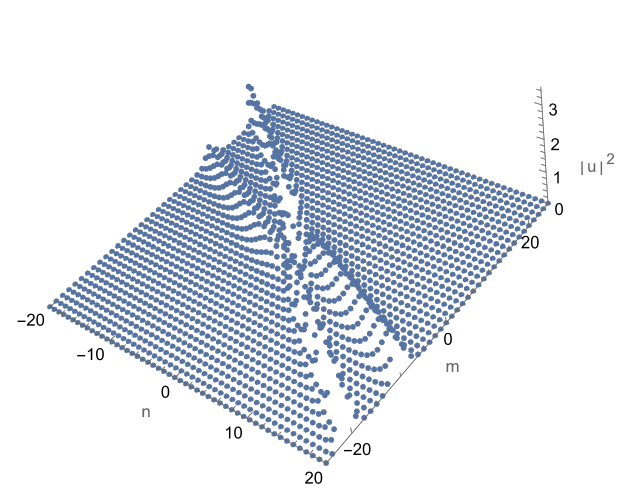}}}
\put(0,-23){\resizebox{!}{4.0cm}{\includegraphics{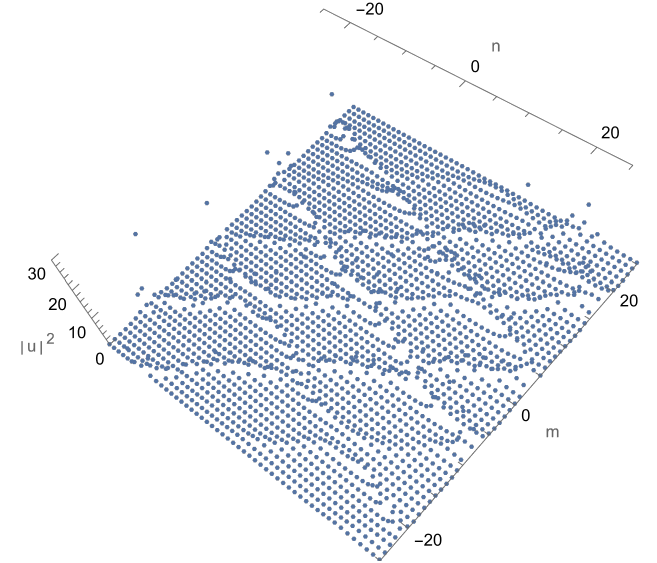}}}
\put(160,-23){\resizebox{!}{4.0cm}{\includegraphics{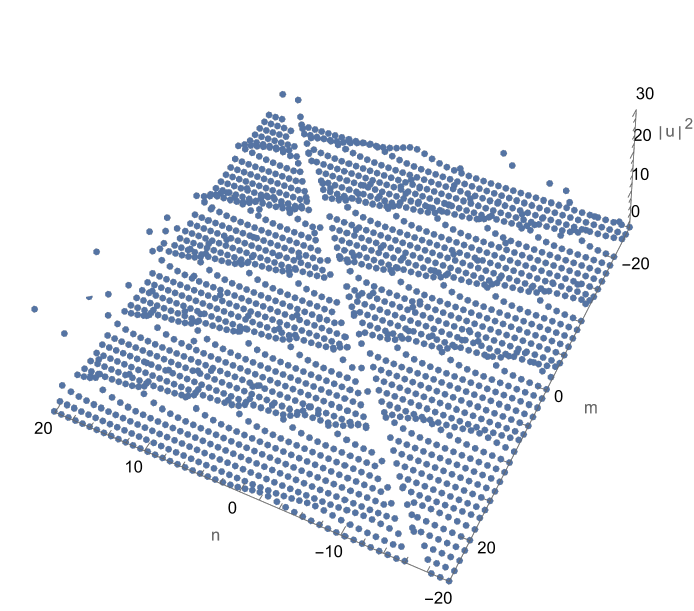}}}
\end{picture}
\end{center}
\vskip 20pt
\begin{center}
\begin{minipage}{15cm}{\footnotesize
\qquad\qquad\qquad(a)\qquad\qquad\qquad\qquad\qquad\qquad\qquad\qquad\quad(b)\qquad\qquad\qquad\qquad\qquad\qquad\qquad\quad (c) \\
{\bf Figure 10.}
Solution $|u|^2$ given by \eqref{ndNLS-fg-2ss}:
(a) Bell two-soliton solution with $k_1=0.4+0.01i$, $k_2=0.6+0.02i$ and $\delta=-1$; (b)
Oscillatory two-soliton solution with $k_1=0.01+0.44i$, $k_2=0.002+0.19i$ and $\delta=-0.4$;
(c) Bell-Oscillatory two-soliton solution with $k_1=0.005+0.2i$, $k_2=3+0.3i$ and $\delta=-1$.}
\end{minipage}
\end{center}

Furthermore, solution \eqref{ndNLS-fg-2ss} with $k_2=ik_1$ yields the breather solution, which is illustrated in Figure 11.
\vskip35pt
\begin{center}
\begin{picture}(120,80)
\put(-120,-23){\resizebox{!}{4.0cm}{\includegraphics{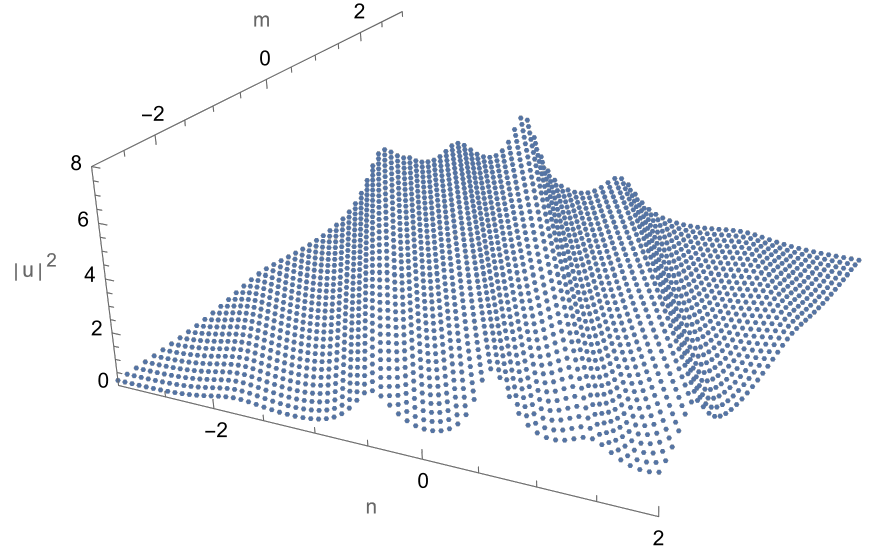}}}
\put(100,-23){\resizebox{!}{4.0cm}{\includegraphics{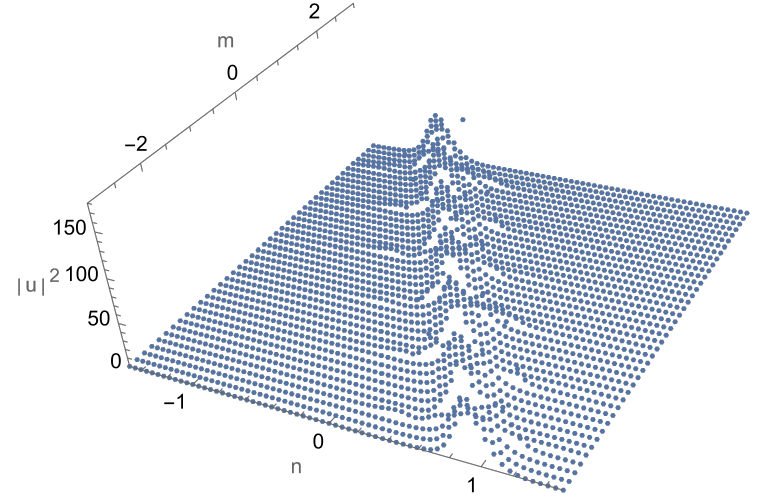}}}
\end{picture}
\end{center}
\vskip 20pt
\begin{center}
\begin{minipage}{15cm}{\footnotesize
\qquad\qquad\qquad\qquad\qquad\qquad\quad(a)\qquad\qquad\qquad\qquad\qquad\qquad\qquad\qquad\qquad\qquad (b) \\
{\bf Figure 11.} Breather solution $|u|^2$ given by \eqref{ndNLS-fg-2ss} with $\delta=-4$:
(a) shape and movement with $k_1=0.55+0.1i$ and $k_2=-0.1+0.55i$; (b) shape and movement with $k_1=1.45+0.1i$
and $k_2=-0.1+1.45i$.}
\end{minipage}
\end{center}

\subsubsection{Jordan-block solutions}

With the Jordan-block matrix $L$ \eqref{dNLS-L-Jor} and $C^{+}=\breve{I}$, solution $\Phi$ is composed by
\begin{align*}
\Phi_{j}=\left\{
\begin{aligned}
&\dfrac{\partial^{j-1}_{k_1}e^{\xi_1}}{(j-1)!}, \quad j=1, 2, \dots, N+1,\\
&\dfrac{\partial^{s-1}_{k_1}e^{\theta_1}}{(s-1)!}, \quad j=N+1+s,\quad s=1,2,\dots,N+1.
\end{aligned}	
\right.
\end{align*}
When $N=1$, the simplest Jordan-block solution is \eqref{dNLS-uw-fg} with
\begin{subequations}
\label{ndNLS-fg-JB-semi}
\begin{align}
f& =8(1+\cosh(2\xi_1-2\theta_1)-2\eta_1\varrho_1\sinh^22k_1),\\
g& =16\sinh2k_1\big((e^{2\xi_1}+e^{2\theta_1})\cosh2k_1+(\eta_1e^{2\theta_1}+\varrho_1e^{2\xi_1})\sinh2k_1\big),
\end{align}
\end{subequations}
where $\eta_{1}$ is defined by \eqref{dNLS-JBS-fg} and $\varrho_{1}=\partial_{k_1}\theta_{1}$.
This solution is depicted in Figure 12.
\vskip35pt
\begin{center}
\begin{picture}(120,80)
\put(-130,-23){\resizebox{!}{5.5cm}{\includegraphics{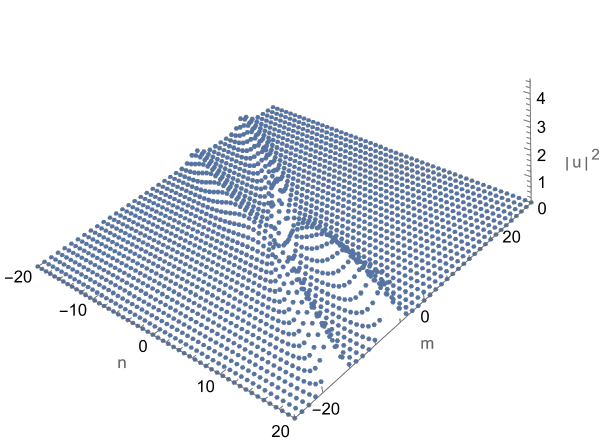}}}
\put(100,-23){\resizebox{!}{4.0cm}{\includegraphics{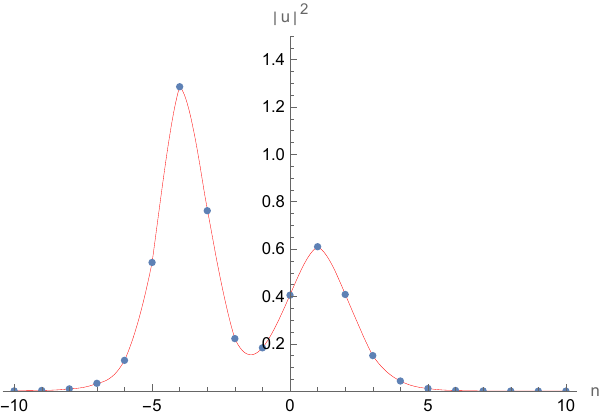}}}
\end{picture}
\end{center}
\vskip 20pt
\begin{center}
\begin{minipage}{15cm}{\footnotesize
\qquad\qquad\qquad\qquad\qquad\qquad\quad(a)\qquad\qquad\qquad\qquad\qquad\qquad\qquad\qquad\qquad\qquad\quad (b) \\
{\bf Figure 12.} Jordan-block solution $|u|^2$ given by \eqref{ndNLS-fg-JB-semi} with $k_1=0.4+0.015i$ and $\delta=-1$:
(a) shape and movement; (b) 2D-plot of (a) at $m=5$.}
	\end{minipage}
\end{center}

\subsection{Continuum limits}
\label{sec:ndNLS-cl}

Details of the two types of limit adopted in Subsection \ref{sec:dNLS-cl} can be also performed on the ndNLS equation \eqref{ndNLS}.

\subsubsection{Semi-continuum limits}

The semi-continuous limit
\begin{align}
\delta\rightarrow 0, \quad m\rightarrow \infty, \quad \text{while keeping} \quad z=m\delta \quad \text{finite},
\end{align}
enables the ndNLS equation \eqref{ndNLS} to yield the nonlocal semi-discrete NLS (nsdNLS) equation
\begin{align}
\label{nsdNLS-z}
 iu_z-2u+(1+uu_{-1})(\wt{u}+\dt{u})=0,
\end{align}
whose solutions can be summarized as follows.
\begin{Thm}
The function \eqref{sdNLS-u-fg} with
\begin{align}
\label{nsdNLS-z-solu}
f=|e^{-N\Ga}\wh{\Phi^{(N)}};e^{N\Ga}T\wh{\Phi_{-1}^{(-N)}}|,\quad
g=|e^{-N\Ga}\widehat{\Phi^{(N+1)}};e^{N\Ga}T\wh{\Phi_{-1}^{(-N+1)}}|,
\end{align}
solve the equation \eqref{nsdNLS-z}, where the $2(N+1)$-th order column vector $\Phi=e^{\Ga n+2iz\sinh^2\Ga }C^{+}$
and $T$ is a constant matrix of order $2(N+1)$ satisfying matrix equations
\begin{align}
\label{nsdNLS-z-Ga-T}
\Ga T- T\Ga=\bm 0,\quad T^{2}=|e^{\Ga}|^{2}I.
\end{align}
\end{Thm}
Here $\Ga$ and $T$ are still take of form \eqref{ndNLS-Ga-T-ex}. Matrices $L$ \eqref{L-diag} and \eqref{dNLS-L-Jor} give rise
to soliton solutions and Jordan-block solutions, respectively.

\noindent {\it Soliton solutions:} With the diagonal matrix \eqref{L-diag}, the basic column vector $\Phi$ is
\begin{align}
\Phi_{j}=\left\{
\begin{aligned}
& e^{\lambda_j},\quad j=1, 2, \ldots, N+1,\\
& e^{\zeta_s},\quad j=N+1+s, \quad s=1,2,\ldots,N+1,
\end{aligned}	
\right.
\end{align}
where $\zeta_s=\lambda_j|_{k_j\longrightarrow-k_s}$.
When $N=0$, equation \eqref{nsdNLS-z} has one-soliton solution (see Eq. (4b) in \cite{DLZ-AMC})
\begin{align}
\label{nsdNLS-z-1ss}
u=e^{\lambda_1+\zeta_1}\sinh(2k_1)\sech(\lambda_1-\zeta_1),
\end{align}
whose modulus is
\begin{align}
\label{nsdNLS-z-1ss-mu}
|u|^2=e^{-4z\sinh a_1\sin b_1}\dfrac{\cosh2a_1-\cos2b_1}{\cosh2a_1n+\cos2b_1n}.
\end{align}
This carrier wave is stationary with top trace
\begin{align}
a_1\sinh2a_1n-b_1\sin2b_1n=0,
\end{align}
and exhibits both of bell and oscillatory structure.
When $a_1$ is not too small, the wave \eqref{nsdNLS-z-1ss-mu}
provides a bell soliton, while when $a_1$ is near to zero, this wave
shows the oscillatory structure due to the cosine in the denominator.
The bell and oscillatory waves are, respectively, depicted in Figures 13(b) and 13(c).
In addition, for any $n$, it is easy to know that
$|u|^2$ approaches to zero as either $(\sinh a_1\sin b_1>0, z\rightarrow +\infty)$
or $(\sinh a_1\sin b_1<0, z\rightarrow -\infty)$, as depicted in Figure 13(a).

\vskip20pt
\begin{center}
\begin{picture}(120,80)
\put(-160,-23){\resizebox{!}{3.5cm}{\includegraphics{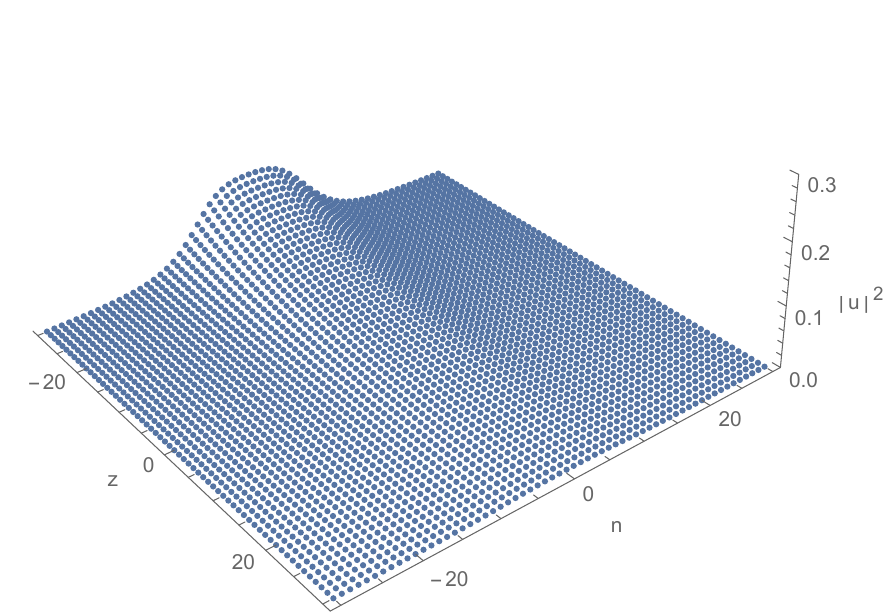}}}
\put(0,-23){\resizebox{!}{3.5cm}{\includegraphics{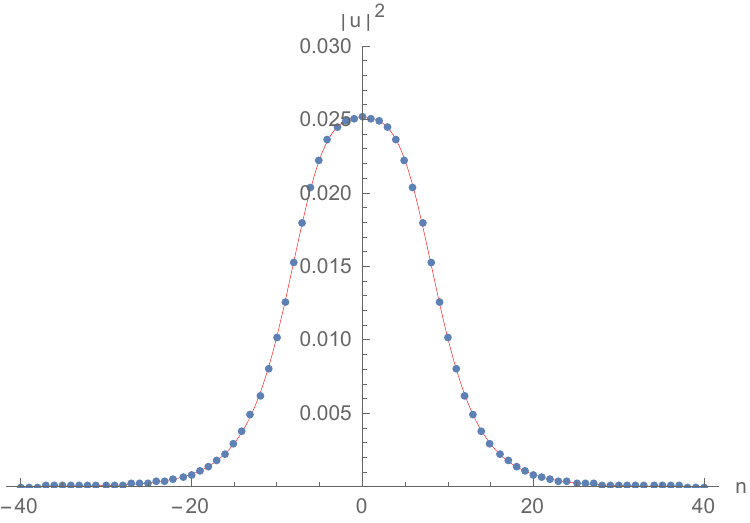}}}
\put(150,-23){\resizebox{!}{4.0cm}{\includegraphics{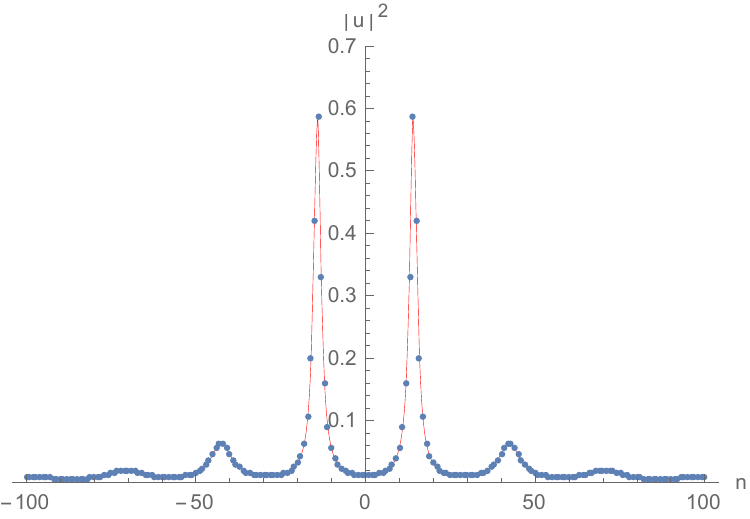}}}
\end{picture}
\end{center}
\vskip 20pt
\begin{center}
\begin{minipage}{15cm}{\footnotesize
\qquad\qquad\qquad(a)\qquad\qquad\qquad\qquad\qquad\qquad\qquad\qquad (b)\qquad\qquad\qquad\qquad\qquad\qquad\qquad\qquad\quad (c)\\
{\bf Figure 13}. One-soliton solution $|u|^2$ given by \eqref{nsdNLS-z-1ss-mu}:
(a) Bell soliton solution with $k_1=0.06+0.055i$;
(b) 2D-plot of (a) at $z=1$;
(c) Oscillatory soliton solution with $k_1=0.005+0.055i$ at $z=1$.}
\end{minipage}
\end{center}

When $N=1$, equation \eqref{nsdNLS-z} has two-soliton solution \eqref{sdNLS-u-fg},
where $f$ and $g$ are defined by \eqref{ndNLS-fg-2ss} with $\xi_j\rightarrow \lambda_j$ and $\theta_j\rightarrow\zeta_j~(j=1,2)$.
The two-soliton solution is shown in Figure 14.
\vskip0pt
\begin{center}
\begin{picture}(120,80)
\put(-160,-23){\resizebox{!}{3.5cm}{\includegraphics{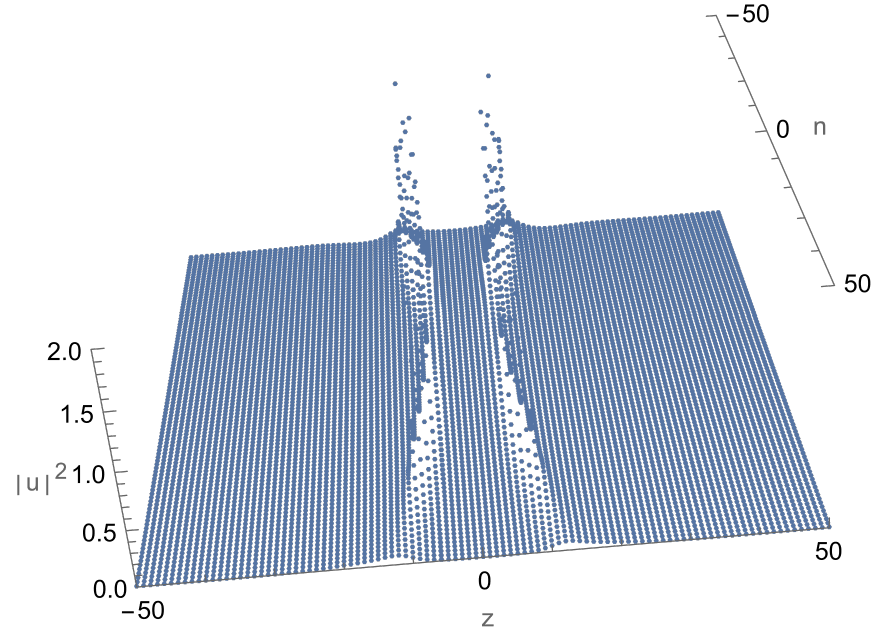}}}
\put(0,-23){\resizebox{!}{3.0cm}{\includegraphics{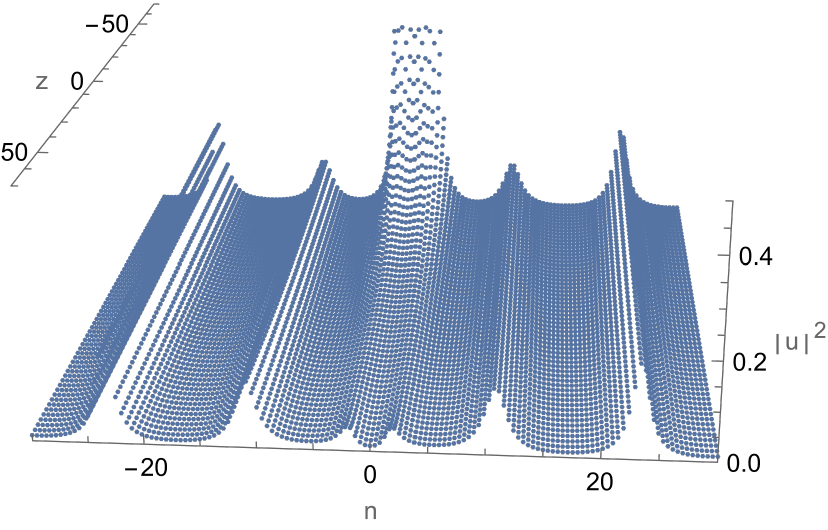}}}
\put(150,-23){\resizebox{!}{3.5cm}{\includegraphics{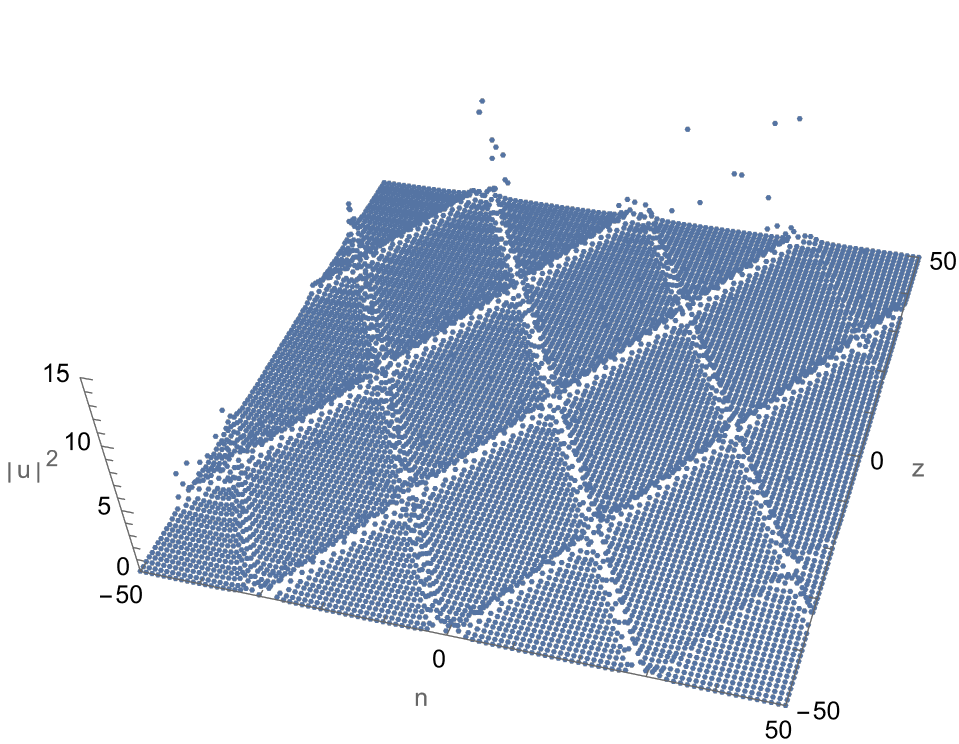}}}
\end{picture}
\end{center}
\vskip 20pt
\begin{center}
\begin{minipage}{15cm}{\footnotesize
\qquad\qquad\qquad(a)\qquad\qquad\qquad\qquad\qquad\qquad\qquad\quad (b)\qquad\qquad\qquad\qquad\qquad\qquad\qquad\qquad\quad (c)\\
{\bf Figure 14.} Two-soliton solution for the equation \eqref{nsdNLS-z}:
(a) Bell two-soliton solution with $k_1=0.01+0.07i$ and $k_2=0.02+0.06i$;
(b) Bell-Oscillatory two-soliton solution with $k_1=0.001+0.055i$ and $k_2=0.07+0.025i$;
(c) Oscillatory two-soliton solution with $k_1=0.001+0.56i$ and $k_2=0.002+0.5i$.}
\end{minipage}
\end{center}

\noindent {\it Jordan-block solutions:} Let $L$ be the Jordan-block matrix \eqref{dNLS-L-Jor},
the basic entries $\{\Phi_j\}$ are thereby of the form
\begin{align}
\label{nsdNLS-JB}
\Phi_{j}=\left\{
\begin{aligned}
&\dfrac{\partial^{j-1}_{k_1}e^{\lambda_1}}{(j-1)!},\quad j=1, 2, \ldots, N+1,\\
&\dfrac{\partial^{s-1}_{k_1}e^{\zeta_1}}{(s-1)!}, \quad j=N+1+s, \quad s=1,2,\ldots,N+1.
\end{aligned}	
\right.
\end{align}
In this case, the simplest Jordan-block solution is still given by \eqref{sdNLS-u-fg}
and \eqref{ndNLS-fg-JB-semi} up to replacements of $\xi_1$ by $\lambda_1$ and $\theta_1$ by $\zeta_1$.
The behavior of $|u|^2$ is depicted in Figure 15.
\vskip39pt
\begin{center}
\begin{picture}(120,80)
\put(-160,-23){\resizebox{!}{4.0cm}{\includegraphics{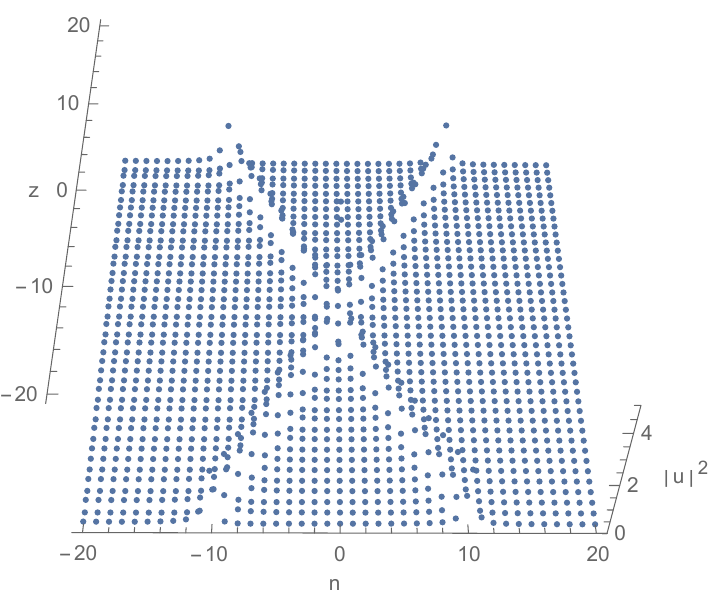}}}
\put(-20,-23){\resizebox{!}{4.0cm}{\includegraphics{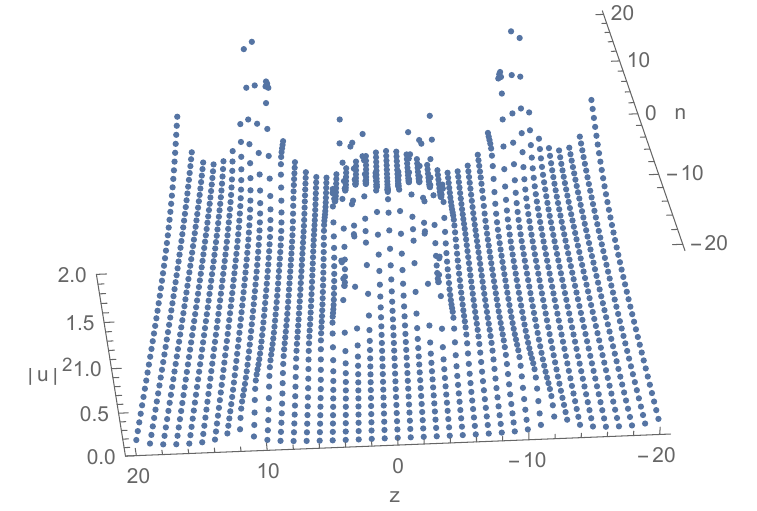}}}
\put(150,-23){\resizebox{!}{3.5cm}{\includegraphics{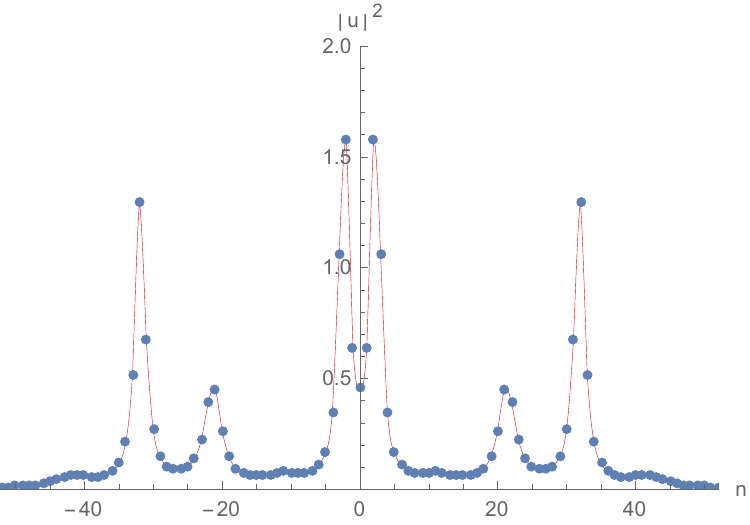}}}
\end{picture}
\end{center}
\vskip 20pt
\begin{center}
\begin{minipage}{15cm}{\footnotesize
\qquad\qquad\qquad(a)\qquad\qquad\qquad\qquad\qquad\qquad\qquad\quad (b)\qquad\qquad\qquad\qquad\qquad\qquad\qquad\quad\quad (c)\\
{\bf Figure 15.} Jordan-block solution $|u|^2$ determined by \eqref{nsdNLS-JB}:
(a) Bell Jordan-block solution with $k_1=0.04+1.4i$; (b)
 Bell-Oscillatory Jordan-block solution with $k_1=0.05+3i$; (c)
2D-plot of (b) at $z=1$.}
\end{minipage}
\end{center}

\subsubsection{Full-continuum limits}

Applying full-continuous limit \eqref{sdNLS-FCM-n}, from the nsdNLS equation \eqref{nsdNLS-z}
one can recover the nonlocal continuous NLS (ncNLS) equation
\begin{align}
\label{ncNLS}
iu_t+u_{xx}+u^2u_{-1}=0,
\end{align}
which has double Wronskian solution,
listed in the following theorem.
\begin{Thm}
The function \eqref{sdNLS-u-fg} with
\begin{align}
\label{ncNLS-solu}
f=|e^{-N\Ga}\wh{\Phi^{(N)}};e^{N\Ga}T\wh{\Phi^{*(-N)}}|,\quad
g=2\sqrt{2}|e^{-N\Ga}\widehat{\Phi^{(N+1)}};e^{N\Ga}T\wh{\Phi^{*(-N+1)}}|,
\end{align}
solves the equation \eqref{ncNLS}, where $\Phi=e^{\Ga x+2i\Ga^2t}C^{+}$
and $T$ is a constant matrix of order $2(N+1)$ satisfying matrix equations
\begin{align}
\label{ncNLS-Ga-T}
\Ga T-T\Ga=\bm 0,\quad T^{2}=|e^{\Ga}|^{2}I.
\end{align}
\end{Thm}

\section{Conclusions}

In this paper we have constructed double Casoratian solutions for the discrete NLS
type equations \eqref{dNLS} and \eqref{ndNLS} by using the bilinearization reduction method.
These two equations are local and nonlocal reductions of the dAKNS(2) equation \eqref{dAKNS2}, respectively.
Starting from the double Casoratian solutions to the dAKNS(2) equation \eqref{dAKNS2}, one-, two-soliton
solutions and the simplest Jordan-block solutions of the resulting equations are determined by
solving the matrix equations. Dynamics of one-soliton and two-soliton solutions
of variable $u$ are given by asymptotic analysis and illustration as a demonstration. By introducing the
semi-continuous limit and the full continuum limit, we recover the discrete NLS
type equations \eqref{dNLS} and \eqref{ndNLS} to their semi-discrete counterparts and
continuous correspondences. The one-, two-soliton
solutions and the simplest Jordan-block solutions for the semi-discrete NLS equations are
discussed. Different from the local case, solutions for nonlocal NLS equations always
exhibit oscillatory structure since the involvement of trigonometric functions.

We finish the paper by the following remarks.

First of all, inserting the semi-continuous limit introduced in Subsubsection \ref{Semi-CL} in the
dAKNS(2) equation \eqref{dAKNS2}, one arrives at the famous Ablowitz-Ladik equation
\begin{subequations}
\label{AL}
\begin{align}
& iu_z-2u+(1+uv)(\wt{u}+\dt{u})=0, \\
& iv_z+2v-(1+uv)(\wt{v}+\dt{v})=0.
\end{align}
\end{subequations}
For the system \eqref{AL} itself, it admits four reductions \cite{DLZ-AMC}
\begin{align}
v=\omega u^*, \quad v=\omega u_{-n}^*, \quad v=\omega u_{-z}, \quad v=\omega u_{-n,-z},
\quad \text{with} \quad \omega=\pm 1,
\end{align}
which correspond to complex local reduction, complex reverse space reduction,
real reverse time reduction, and real reverse space-time reduction, respectively.
While for the dAKNS(2) equation \eqref{dAKNS2}, it just has the
complex local reduction $(v=\omega u^*, w=\omega w^*)$ and the real reverse space-time reduction
$(v=\omega u_{-n,-m},~w=\omega \wh{w}_{-n,-m})$. In the present paper, we take $\omega=1$ without loss of generality.
In fact, the reduced equation with $\omega=1$ and with $\omega=-1$ can be transformed into each other by taking
simple transformations \cite{XZS-SAPM,ZXS-MMAS}.

What is more, for the obtained soliton solutions to the ndNLS equation \eqref{ndNLS}, they can be
bell or oscillatory shape by controlling the values of real parts of $k_i$ (see figures 9, 10).
When $a_i$ are near to zero, oscillatory phenomenon would be arisen since the dominated role of
trigonometric function. While when $a_i$ are not too small, the effect brought by the trigonometric function can be neglected and
the soliton appears as the usual bell form. Therefore, for the two-soliton solution \eqref{ndNLS-fg-2ss},
one can take a small $a_1$ and not small $a_2$ to show the interaction of bell soliton and
oscillatory soliton (see figures 10(c)).

In addition, by discretizing the bilinear form of an extended coupled NLS system,
a generalized discrete NLS system was proposed in \cite{Hirota}, where two-soliton solutions
for the $2$-coupled discrete system were presented. A natural question is how can we
construct the double Casoratian solutions for the general $2M$-coupled discrete NLS system?
Whether the $2M$-coupled discrete NLS system admits local and nonlocal reductions and how can we
construct solutions for the reduced local and nonlocal $M$-coupled NLS system? These are the intriguiging
questions deserved further consideration, which will done in near future.

\vskip 20pt
\section*{Acknowledgments}
This project is supported by the Zhejiang Provincial Natural Science Foundation (No. LZ24A010007)
and National Natural Science Foundation of China (No. 12071432).

\vskip 20pt
\section*{Author Contributions}
S.L.: methodology,  formal analysis, editing, validation,
supervision; X.H.: writing original draft preparation, investigation, formal analysis, validation;
W.: investigation, editing, validation. All authors reviewed the manuscript.

\vskip 20pt
\section*{Data Availibility Statement}
Data sharing not applicable to this article as no datasets were generated or
analyzed during the current study.

\vskip 20pt
\section*{Conflict of interest}
There are no conflicts of interest to declare.

\end{document}